\documentclass[11pt]{article}
\usepackage{graphicx} 

\title{Max-Distance Sparsification for Diversification and Clustering}
\author{Soh Kumabe\thanks{CyberAgent.}}
\date{}
\usepackage{graphicx} 
\usepackage{amsmath,amssymb,amsthm,amsfonts}
\usepackage{bm,fullpage,hyperref}
\usepackage{color}
\usepackage{booktabs,multirow}
\usepackage{comment}
\usepackage{apxproof}

\usepackage{algpseudocode}
\usepackage[linesnumbered,ruled,vlined]{algorithm2e}
\SetKwProg{Procedure}{Procedure}{}{}

\allowdisplaybreaks

\newcommand{\DP}{\mathsf{DP}}
\newcommand{\EX}{\mathsf{EX}}
\newcommand{\OPT}{\mathsf{OPT}}

\newcommand{\floor}[1]{\left \lfloor{#1}\right \rfloor}
\newcommand{\ceil}[1]{\left \lceil{#1}\right \rceil}

\newtheorem{theorem}{Theorem}[section]
\newtheorem{lemma}[theorem]{Lemma}

\theoremstyle{definition}
\newtheorem{definition}[theorem]{Definition}

\begin{document}

\maketitle
\begin{abstract}
Let $\mathcal{D}$ be a set family that is the solution domain of some combinatorial problem.
The \emph{max-min diversification problem on $\mathcal{D}$} is the problem to select $k$ sets from $\mathcal{D}$ such that the Hamming distance between any two selected sets is at least $d$. FPT algorithms parameterized by $k+\ell $, where $\ell=\max_{D\in \mathcal{D}}|D|$, and $k+d$ have been actively studied recently for several specific domains.

This paper provides unified algorithmic frameworks to solve this problem. Specifically, for each parameterization $k+\ell $ and $k+d$, we provide an FPT oracle algorithm for the max-min diversification problem using oracles related to $\mathcal{D}$. 
We then demonstrate that our frameworks provide the first FPT algorithms on several new domains $\mathcal{D}$, including the domain of $t$-linear matroid intersection, almost $2$-SAT, minimum edge $s,t$-flows, vertex sets of $s,t$-mincut, vertex sets of edge bipartization, and Steiner trees.
We also demonstrate that our frameworks generalize most of the existing domain-specific tractability results.

Our main technical breakthrough is introducing the notion of \emph{max-distance sparsifier} of $\mathcal{D}$, a domain on which the max-min diversification problem is equivalent to the same problem on the original domain $\mathcal{D}$.
The core of our framework is to design FPT oracle algorithms that construct a constant-size max-distance sparsifier of $\mathcal{D}$.
Using max-distance sparsifiers, we provide FPT algorithms for the max-min and max-sum diversification problems on $\mathcal{D}$, as well as $k$-center and $k$-sum-of-radii clustering problems on $\mathcal{D}$, which are also natural problems in the context of diversification and have their own interests.
\end{abstract}


\section{Introduction}

\subsection{Background and Motivation}

The procedure for approaching real-world problems with optimization algorithms involves formulating the real-world motivations as mathematical problems and then solving them. However, real-world problems are complex, and the idea of a ``good'' solution cannot always be correctly formulated. 
The paradigm of \emph{diversification}, introduced by Baste et al.~\cite{baste2019fpt} and Baste et al.~\cite{baste2022diversity}, is a ``formulation of the unformulatable problems'', which 
formulates \emph{diversity measures} for a set of multiple solutions, rather than attempting to formulate the ``goodness'' of a single solution. 
By computing a set of solutions that maximize this measure, the algorithm provides effective options to evaluators who have the correct criteria for judging the ``goodness'' of a solution.

Let $U$ be a finite set, $k\in \mathbb{Z}_{\geq 1}$ and $d\in \mathbb{Z}_{\geq 0}$.
Let $\mathcal{D}\subseteq 2^U$ be the feasible domain of some combinatorial problem.
The following problem frameworks, defined by two types of diversity measures, have been studied extensively.

\begin{quote}
\textbf{Max-Min Diversification Problem on $\mathcal{D}$:}  
Does there exist a $k$-tuple $(D_1, \dots, D_k) \in \mathcal{D}^k$ of sets in $\mathcal{D}$ such that $\min_{1\leq i<j\leq k} |D_i \triangle D_j| \geq d$?\footnote{We define $Z_1\triangle Z_2 := (Z_1\setminus Z_2)\cup (Z_2\setminus Z_1)$.}
\end{quote}

\begin{quote}
\textbf{Max-Sum Diversification Problem on $\mathcal{D}$:}  
Does there exist a $k$-tuple $(D_1, \dots, D_k) \in \mathcal{D}^k$ of sets in $\mathcal{D}$ such that $\sum_{1\leq i<j\leq k} |D_i \triangle D_j| \geq d$?
\end{quote}
These problems ensure diversity by aiming to output solutions that are as dissimilar as possible in terms of Hamming distance.

Parameterized algorithms for diversification problems have been actively studied.
Particularly, FPT algorithms for the max-min diversification problems parameterized by $k+\ell $, where $\ell= \max_{D \in \mathcal{D}} |D|$~\cite{baste2022diversity,baste2019fpt,fomin2024diversecollection,hanaka2021finding,shida2024finding}, as well as by $k+d$~\cite{eiben2024determinantal,fomin2024diversepair,fomin2024diversecollection,funayama2024parameterized,gima2024computing,misra2024parameterized}, have been the focus of research.
Since assuming $d \leq 2\ell $ does not lose generality in the max-min diversification problem, the latter addresses a more general situation than the former. 
Since the max-sum diversification problem is empirically more tractable than the max-min diversification problem, for some time hereafter, we will restrict our discussion to the max-min diversification problem.


This research provides general algorithmic frameworks for FPT algorithms solving max-min (and max-sum) diversification problems for both parameterizations $k+\ell $ and $k+d$. 
Our frameworks are very general and can be applied to all domains~\cite{bang2021k,bang2016parameterized,baste2022diversity,baste2019fpt,eiben2024determinantal,fomin2024diversepair,fomin2024diversecollection,funayama2024parameterized,gutin2018k,hanaka2021finding,shida2024finding} for which FPT algorithms parameterized by $k+\ell $ and $k+d$ are currently known for the case that diversity measure is defined using an unweighted Hamming distance. 
Moreover, our frameworks further provide the first algorithms for several domains where such algorithms were previously unknown.


Our main technical breakthrough is introducing a notion of \emph{max-distance sparsifier} as an intermediate step, which, for the max-min diversification problem, essentially works as a \emph{core-set}~\cite{agarwal2004approximating}. The formal definition is given in Section~\ref{sec:framework_overview}.
The critical fact is that, when $\mathcal{K}$ is a max-distance sparsifier of $\mathcal{D}$, the max-min diversification problem on $\mathcal{D}$ is equivalent to the same problem on $\mathcal{K}$.
Our framework constructs a constant-size max-distance sparsifier $\mathcal{K}$ of $\mathcal{D}$ using the oracles on $\mathcal{D}$, enabling us to solve the max-min diversification problems on $\mathcal{D}$ by brute-force search on $\mathcal{K}$.

The power of max-distance sparsification is not limited to solving diversification problems. Specifically, the following \emph{$k$-center}~\cite{hsu1979easy} and \emph{$k$-sum-of-radii clustering problems on $\mathcal{D}$}~\cite{charikar2001clustering} can also be solved via max-distance sparsification.

\begin{quote}
\textbf{$k$-Center Clustering Problem on $\mathcal{D}$:}  
Does there exist a $k$-tuple of subsets $(D_1, \dots, D_k) \in \mathcal{D}^k$ such that for all $D \in \mathcal{D}$, there exists an $i \in \{1, \dots, k\}$ satisfying $|D_i \triangle D| \leq d$?
\end{quote}

\begin{quote}
\textbf{$k$-Sum-of-Radii Clustering Problem on $\mathcal{D}$:}  
Does there exist a $k$-tuple of subsets $(D_1, \dots, D_k) \in \mathcal{D}^k$ and a $k$-tuple of non-negative integers $(d_1, \dots, d_k) \in \mathbb{Z}_{\geq 0}^k$ with $\sum_{i \in \{1, \dots, k\}} d_i \leq d$ such that for all $D \in \mathcal{D}$, there exists an $i \in \{1, \dots, k\}$ satisfying $|D_i \triangle D| \leq d_i$?
\end{quote}

When $\mathcal{D}$ is an explicitly given set of points, parameterized algorithms for these problems have been extensively studied in the area of clustering~\cite{amir2014efficiency,bandyapadhyayL023a, chen2024parameterized, demaine2005fixed, eiben2023parameterized, feldmann2020parameterized, inamdar2020capacitated}.
Furthermore, approximation algorithms for the \emph{relational $k$-means}~\cite{curtin2020rk,EsmailpourS24,MoseleyPSW21} and \emph{relational $k$-center}~\cite{agarwal2024computing} problems are investigated, which are the $k$-means and $k$-center clustering problem defined on a point set represented as a \emph{join} of given relational databases. 
Their setting is similar to ours as the point set are implicitly given and its size can be exponential. 
This research adds several new combinatorial domains to the literature on clustering problems on implicitly given domains, and also introduces a parameterized view.
These problems are also natural in the context of diversification, since in real situations, the concept of diversity often means that the extracted elements cover the entire space comprehensively rather than being mutually dissimilar. This motivation is formulated by clustering problems, which extract a list of sets in $\mathcal{D}$ such that for each set in $\mathcal{D}$, there is an extracted set near to it.

\subsection{Our Results}

This paper consists of two parts. In the first part, we design general frameworks for solving diversification and clustering problems. In the second part, we apply our frameworks to several specific domains $\mathcal{D}$. 
Table 1 shows a list of domains on which our frameworks provide the first or an improved algorithm for the max-min diversification problem.

\subsubsection{The Frameworks}

We define the following \emph{$(-1,1)$-optimization oracle on $\mathcal{D}$} and the \emph{exact extension oracle on $\mathcal{D}$}.

\begin{quote}
\textbf{$(-1,1)$-Optimization Oracle on $\mathcal{D}$:}  
Let $U$ be a finite set, $\mathcal{D}\subseteq 2^U$, and $w\in \{-1,1\}^U$ be a weight vector. Return a set $D\in \mathcal{D}$ that maximizes $\sum_{e\in D}w_e$.
\end{quote}

\begin{quote}
\textbf{Exact Extension Oracle on $\mathcal{D}$:}  
Let $U$ be a finite set, $r\in \mathbb{Z}_{\geq 0}$, $\mathcal{D}\subseteq 2^U$, and $C\in \mathcal{D}$. Let $X,Y\subseteq U$ be two disjoint subsets of $U$. If there exists a set $D\in \mathcal{D}$ such that $|D\triangle C|=r$, $X\subseteq D$, and $Y\cap D = \emptyset$, return one such set. If no such set exists, return $\bot$.
\end{quote}

When $C=\emptyset$ and $X=\emptyset$, we specifically call the exact extension oracle on $\mathcal{D}$ the \emph{exact empty extension oracle on $\mathcal{D}$}.

\begin{quote}
\textbf{Exact Empty Extension Oracle on $\mathcal{D}$:}  
Let $U$ be a finite set, $r\in \mathbb{Z}_{\geq 0}$, $\mathcal{D}\subseteq 2^U$, and $Y\subseteq U$. If there exists a set $D\in \mathcal{D}$ such that $|D|=r$ and $Y\cap D = \emptyset$, return one such set. If no such set exists, return $\bot$.
\end{quote}

Let $\mathcal{P}_{\mathcal{D}}$ be the max-min/max-sum diversification problem or the $k$-center/$k$-sum-of-radii clustering problem on $\mathcal{D}$.
Our main result is FPT algorithms for solving $\mathcal{P}_{\mathcal{D}}$ using these oracles. 
We construct frameworks for both types of parameterizations, $k+\ell $ and $k+d$. 
The result for the parameterization by $k+\ell $ is as follows.
\begin{theorem}\label{thm:framework_small}
There exists an oracle algorithm solving $\mathcal{P}_{\mathcal{D}}$ using the exact empty extension oracle on $\mathcal{D}$, where the number of oracle calls and time complexity are both FPT parameterized by $k+\ell $ and for each call of an oracle, $r$ and $|Y|$ are bounded by constants that depend only on $k+\ell $.
\end{theorem}

The result for the parameterization by $k+d$ is as follows.
\begin{theorem}\label{thm:framework_large}
There exists a randomized oracle algorithm solving $\mathcal{P}_{\mathcal{D}}$ using the $(-1,1)$-optimization oracle on $\mathcal{D}$ and the exact extension oracle on $\mathcal{D}$, where the number of oracle calls and time complexity are both FPT parameterized by $k+d$ and for each call of the exact extension oracle, $r+|X|+|Y|$ are bounded by constants that depend only on $k+d$.
\end{theorem}

\begin{table}[t!]
    \centering
    \begin{tabular}{c|c|c}
        Domain & Parameter & Reference \\ \hline \hline
        $t$-Linear Matroid Intersection & $k+\ell+t$ & Sec.~\ref{sec:app_matroid_intersection}\\ \hline
        Almost $2$-SAT  & $k+\ell $ & Sec.~\ref{sec:app_almostsat}\\ \hline
        Independent Set on Certain Graphs  & $k+\ell $ & Sec.~\ref{sec:app_mis}\\ \hline
        Min Edge $s,t$-Flow  & $k+d$ &Sec.~\ref{sec:app_flow}\\ \hline
        Steiner Tree  & $k+d+|T|$ &Sec.~\ref{sec:app_steiner}\\ \hline
        Vertex Set of Min $s,t$-Cut & $k+d$ &Sec.~\ref{sec:app_cut}\\ \hline
        Vertex Set of Edge Bipartization & $k+d+s$ &Sec.~\ref{sec:app_bipartization}\\ \hline
    \end{tabular}
    \caption{List of new results for the max-min diversification problem obtained by our frameworks. 
    The first column represents the domain $\mathcal{D}$. 
    The second and third columns represent parameterization and the reference, respectively.
    For the formal definition of each domain, see Section~\ref{sec:application}.}
\end{table}

\subsubsection{Applications of Theorem~\ref{thm:framework_small}}\label{sec:intro_app_small}

On most domains $\mathcal{D}$, the exact empty extension oracle can be designed by almost the same way as an algorithm to extract a single solution from $\mathcal{D}$. 
For example, consider the case where $\mathcal{D}$ is the \emph{$\ell $-path domain}, i.e., a domain of sets of edges on paths of length $\ell $. 
In this case, the exact empty extension oracle on $\mathcal{D}$ is equivalent to the problem of finding an $\ell $-path in the graph obtained by removing all edges in $Y$ from the input.
Combining Theorem~\ref{thm:framework_small} with this empirical fact, we can claim that, for most domains $\mathcal{D}$, the diversification and clustering problems parameterized by $k+\ell $ on $\mathcal{D}$ are as easy as determining the non-emptiness of $\mathcal{D}$.

To demonstrate that Theorem~\ref{thm:framework_small} yields existing tractability results, we design the oracles for the domains of the vertex cover~\cite{baste2022diversity,baste2019fpt}, $t$-hitting set~\cite{baste2022diversity}, feedback vertex set~\cite{baste2022diversity}, and common independent set of two matroids~\cite{eiben2024determinantal,fomin2024diversecollection}. We also apply our framework on new domains, $t$-represented linear matroid intersection, almost $2$-SAT, and independent set on \emph{subgraph-closed IS-FPT} graph classes. Here, a graph class is \emph{subgraph-closed IS-FPT} if it is closed under taking subgraphs and the problem of finding independent set of size $\ell $ is FPT parameterized by $\ell $.
The following theorem summarizes our results, where the precise definitions of each domain are given in Section~\ref{sec:paramkl}.
\begin{theorem}
Let $\mathcal{D}$ be the domain of vertex covers, $t$-hitting sets, feedback vertex sets, $t$-represented linear matroid intersections, almost $2$-SATs, or independent sets on subgraph-closed IS-FPT graph classes.
Then, max-min and max-sum diversification problems and $k$-center and $k$-sum-of-radii clustering problems admit an FPT algorithm, where the parameterization is $k+\ell $ except for the $t$-hitting set and $t$-represented linear matroid intersection, which are parameterized by $k+\ell+t$.
\end{theorem}

Theorem~\ref{thm:framework_small} also generalizes existing frameworks for diversification. Baste et al.~\cite{baste2019fpt} provided an algorithmic framework for diversification using a \emph{loss-less kernel}~\cite{carbonnel2016propagation,carbonnel2017kernelization}, which, roughly speaking, is a kernel that completely preserves the information of the solution space. Since loss-less kernels are known for very limited domains, their framework requires very strong assumptions. Our framework has broader applicability than theirs because it relies on a weaker oracle, as the exact empty extension oracle can be constructed using a loss-less kernel. 
Hanaka et al.~\cite{hanaka2021finding} developed a color-coding-based framework for diversification. The oracle they use can be regarded as the exact empty extension oracle with an additional colorfulness constraint. Our framework again has broader applicability than theirs because our oracle can be constructed using theirs. Moreover, our framework also treats clustering problems, which these two do not.

\subsubsection{Applications of Theorem~\ref{thm:framework_large}}\label{sec:intro_app_large}

FPT algorithms for the max-min diversification problems on $\mathcal{D}$ parameterized by $k+d$ are known for the cases where $\mathcal{D}$ is the family of matroid bases~\cite{eiben2024determinantal,fomin2024diversecollection}, perfect matchings~\cite{fomin2024diversecollection}, and shortest paths~\cite{funayama2024parameterized}. 
The result for the perfect matchings is later extended to the matchings of specified size, not necessarily perfect~\cite{eiben2024determinantal}.
Additionally, for the cases where $\mathcal{D}$ is the family of interval schedulings~\cite{hanaka2021finding} and the longest common subsequences of an absolute constant number of strings~\cite{shida2024finding}, FPT algorithms parameterized by $k+\ell $ are known. Both of these two can be generalized to the domain of dynamic programming problems, which we define in Section~\ref{sec:paramkd}. 
Furthermore, the problem of finding a pair of a branching and an in-branching such that the Hamming distance between them is at least $d$ is investigated as the name of \emph{$d$-distinct branchings problem}~\cite{bang2021k,bang2016parameterized,gutin2018k}, which admits FPT algorithm parameterized by $d$.
This problem can naturally be extended to the case that selects $k_1$ branchings and $k_2$ in-branching, rather than one each.
We give FPT algorithms parameterized by $k+d$ for all those problems, where $k=k_1+k_2$ for the extended version of $d$-distinct branchings problem.
We also give FPT algorithms on domains of minimum edge $s,t$-flows, Steiner trees, vertex sets of $s,t$-mincut, and vertex sets of edge bipartization, which are domains where no FPT algorithm for the max-min diversification problem is previously known.
Remark that the domain of shortest paths~\cite{funayama2024parameterized} is the special case of the minimum edge $s,t$-flow domain and the minimum Steiner tree domain. 
The following theorem summarizes our results, where the precise definitions of each domain are given in Section~\ref{sec:paramkd}.
\begin{theorem}
Let $\mathcal{D}$ be the domain of matroid bases, branchings, matchings of specified size, minimum edge $s,t$-flows, minimum Steiner trees, vertex sets of $s,t$-mincut, vertex sets of edge bipartization, and dynamic programming problems.
Then, max-min and max-sum diversification problems and $k$-center and $k$-sum-of-radii clustering problems admit an FPT algorithm, where the parameterization is $k+d$ except for the Steiner tree, which is parameterized by $k+d+|T|$ for the terminal set $T$, and edge bipartization, which is parameterized by $k+\ell+s$, where $s$ is the minimum number of edges to be removed to make the given graph bipartite.
Furthermore, the extended version of $d$-distinct branching problem also admits FPT algorithm parameterized by $k+d$.
\end{theorem}

Eiben et al.~\cite{eiben2024determinantal} provided a technique called \emph{determinantal sieving}, which is a general tool to give and speed up parameterized algorithms, including that for diversification problems. Particularly, they provided a framework to solve the diversification problem by using an oracle that, roughly speaking, counts the number of solutions modulo $2$. 
Using their framework, they improved the running times of FPT algorithms for max-min diversification problems on matchings and matroid bases, as well as the $d$-distinct branchings problem.
Although not stated explicitly, their framework seems to yield FPT algorithms parameterized by $k+d$ when $\mathcal{D}$ is the dynamic programming domain, thereby improving the parameterizations in the results of~\cite{hanaka2021finding} and~\cite{shida2024finding}, respectively, as well as the extended version of $d$-distinct branchings problems.
We are not sure whether our framework generalizes theirs, that is, whether our oracle can be constructed using their oracle.  
However, we strongly believe that our framework has broader applicability because their framework assumes counting oracles, which is often hard even modulo $2$.
In contrast, our framework uses optimization-type oracles, which are generally more tractable than counting. 
Indeed, our framework provides an FPT algorithm with the same parameterization for every domain which they explicitly considered.
Moreover, we do not think their framework can give an FPT algorithm for the max-min diversification problem on the domains of minimum edge $s,t$-flows, vertex sets of minimum $s,t$-cut, and vertex sets of edge bipartization.
Furthermore, our framework can be applied not only to diversification problems but also to clustering problems.

\subsection{Framework Overview}\label{sec:framework_overview}
In this section, we provide an overview of the entire flow of our frameworks. Our frameworks first construct max-distance sparsifiers using the corresponding oracles and then solve diversification and clustering problems using them. 

\subsubsection{From $d$-Limited $k$-Max-Distance Sparsifier to Diversification and Clustering}
We begin by defining the max-distance sparsifiers.
The key for Theorem~\ref{thm:framework_small} is designing the following \emph{$k$-max-distance sparsifier of $\mathcal{D}$}.
\begin{definition}[\rm{$k$-max-distance sparsifier}]
Let $k \in \mathbb{Z}_{\geq 1}$. Let $U$ be a finite set and $\mathcal{D}, \mathcal{F} \subseteq 2^U$. We say that $\mathcal{K} \subseteq \mathcal{D}$ is a \emph{$k$-max-distance sparsifier of $\mathcal{D}$ with respect to $\mathcal{F}$} if for any $(F_1, \dots, F_k) \in \mathcal{F}^k$ and $(z_1, \dots, z_k) \in \mathbb{Z}_{\geq 0}^{k}$, the two conditions
\begin{itemize}
    \item There exists $D \in \mathcal{D}$ such that for each $i\in \{1, \dots, k\}$, $|F_i \triangle D| \geq z_i$.
    \item There exists $K \in \mathcal{K}$ such that for each $i\in \{1, \dots, k\}$, $|F_i \triangle K| \geq z_i$.
\end{itemize}
are equivalent.
Unless specifically noted, when we write \emph{$k$-max-distance sparsifier of $\mathcal{D}$}, we mean the case where $\mathcal{D} = \mathcal{F}$.
\end{definition}

Similarly, the key for Theorem~\ref{thm:framework_large} is designing the following \emph{$d$-limited $k$-max-distance sparsifier of $\mathcal{D}$}.
\begin{definition}[\rm{$d$-limited $k$-max-distance sparsifier}]
Let $k \in \mathbb{Z}_{\geq 1}$ and $d \in \mathbb{Z}_{\geq 0}$. Let $U$ be a finite set and $\mathcal{D}, \mathcal{F} \subseteq 2^U$. We say that $\mathcal{K} \subseteq \mathcal{D}$ is a \emph{$d$-limited $k$-max-distance sparsifier of $\mathcal{D}$ with respect to $\mathcal{F}$} if for any $(F_1, \dots, F_k) \in \mathcal{F}^k$ and $(z_1, \dots, z_k) \in \{0,\dots, d\}^k$, the two conditions
\begin{itemize}
    \item There exists $D \in \mathcal{D}$ such that for each $i\in \{1, \dots, k\}$, $|F_i \triangle D| \geq z_i$.
    \item There exists $K \in \mathcal{K}$ such that for each $i\in \{1, \dots, k\}$, $|F_i \triangle K| \geq z_i$.
\end{itemize}
are equivalent.
Unless specifically noted, when we write \emph{$d$-limited $k$-max-distance sparsifier of $\mathcal{D}$}, we mean the case where $\mathcal{D} = \mathcal{F}$.
\end{definition}

The difference between the two sparsifiers is that the domain of $(z_1,\dots, z_k)$ is $\mathbb{Z}_{\geq 0}^k$ in the former case, while it is $\{0,\dots, d\}^k$ in the latter. 
By definition, any $k$-max-distance sparsifier is also a $d$-limited $k$-max-distance sparsifier for any $d\in \mathbb{Z}_{\geq 0}$.
We can prove that given a $d$-limited $(k-1)$-max-distance sparsifier of $\mathcal{D}$ with size bounded by a constant that depends only on $k+d$, we can construct FPT algorithms parameterized by $k+d$ for the max-min/max-sum diversification problems on $\mathcal{D}$. 
Similarly, we can prove that given a $(d+1)$-limited $k$-max-distance sparsifier of $\mathcal{D}$ with size bounded by a constant that depends only on $k+d$, we can construct FPT algorithms parameterized by $k+d$ for the $k$-center/$k$-sum-of-radii clustering problems on $\mathcal{D}$. 
Therefore, to prove Theorems~\ref{thm:framework_small}~and~\ref{thm:framework_large}, it suffices to construct FPT oracle algorithms for designing $k$-max-distance sparsifiers and $d$-limited $k$-max-distance sparsifiers, respectively.

\subsubsection{Computing $k$-Max-Distance Sparsifier}\label{sec:overview_small}

The remaining task towards Theorem~\ref{thm:framework_small} is to provide an FPT algorithm parameterized by $k+\ell $ that constructs a $k$-max-distance sparsifier of $\mathcal{D}$ with size bounded by a constant that depends only on $k+\ell $. 
The key lemma toward this is that, if $\mathcal{K}$ contains a sufficiently large \emph{sunflower} (see Section~\ref{sec:framework_small} for the definition) consisting of sets of the same size, then we can safely remove one of them from $\mathcal{K}$ while preserving the property that $\mathcal{K}$ is a $k$-max-distance sparsifier (actually, for the sake of simplifying the framework, we prove a slightly stronger statement). 
Starting with $\mathcal{K}=\mathcal{D}$ and exhaustively removing such sets leads to a $\mathcal{K}$ that is still a $k$-max-distance sparsifier of $\mathcal{D}$, and by using the well-known sunflower lemma (Lemma~\ref{lem:sunflower}), its size is bounded by a constant.

However, this observation is still not sufficient to obtain an FPT algorithm. 
The reason is that $\mathcal{D}$ generally has exponential size, and removing sets one by one would require an exponential number of steps. Instead, our algorithm starts with $\mathcal{K}=\emptyset$ and exhaustively adds sets of $\mathcal{D}$ to $\mathcal{K}$ until $\mathcal{K}$ becomes a $k$-max-distance sparsifier. In this way, the number of steps is bounded by a constant. 
The remaining task is to choose a set to be added at each step. For this task, we design an FPT algorithm using constant number of calls of the exact empty extension oracle.



\subsubsection{Computing $d$-Limited $k$-Max-Distance Sparsifier}

The algorithm in the previous section alone is insufficient to prove Theorem~\ref{thm:framework_large} since $\ell $ is unbounded and the sunflower-lemma-based bound for the number of steps cannot be used.
Our algorithm divides $\mathcal{D}$ into at most $k$ clusters, computes a $d$-limited $k$-max-distance sparsifier for each cluster, and outputs their union.
Let $p>2d$ be a constant that depends only on $k+d$. 
We first find $\mathcal{C}\subseteq \mathcal{D}$ satisfying the following properties: (i) $|\mathcal{C}|\leq k$, (ii) for all distinct $C,C'\in \mathcal{C}$, $|C\triangle C'| > 2d$, and (iii) for all $D\in \mathcal{D}$, there exists $C\in \mathcal{C}$ such that $|D\triangle C|\leq p$.
If such a family does not exist, a trivial $d$-limited $k$-max-distance sparsifier of $\mathcal{D}$ will be found, and we output it and terminate.

We provide an algorithm for computing $\mathcal{C}$. Our algorithm starts with $\mathcal{C}=\emptyset$ and exhaustively adds sets in $\mathcal{D}$ to $\mathcal{C}$ until $\mathcal{C}$ satisfies the above conditions or its size exceeds $k$.
To choose the elements to be added, we randomly sample $w\in \{-1,1\}^U$ and call a $(-1,1)$-optimization oracle. We can prove for sufficiently large constant $p$ that if there exists $D\in \mathcal{D}$ such that $|D\triangle C|>p$ for any $C\in \mathcal{C}$, with a constant probability, the $(-1,1)$-optimization oracle will find a $D\in \mathcal{D}$ such that $|D\triangle C|>2d$ for any $C\in \mathcal{C}$. 
Thus, if $\mathcal{C}$ does not meet the conditions, by calling the $(-1,1)$-optimization oracle a sufficient number of times, we can find a set to add to $\mathcal{C}$ with high probability.

Here, we provide an algorithm for computing a $d$-limited $k$-max-distance sparsifier of $\mathcal{D}$ using $\mathcal{C}$.
For each cluster $\mathcal{D}_{C} := \{D\in \mathcal{D}\colon |D\triangle C|\leq p\}$, let $\mathcal{D}^*_{C} := \{D\triangle C\colon D\in \mathcal{D}_{C}\}$.
The algorithm computes a $k$-max-distance sparsifier of each $\mathcal{D}^*_{C}$ and outputs their union.
For technical reasons, we actually compute a slightly more general object, but we will not delve into the details here.
Since each $\mathcal{D}^*_{C}$ consists only of sets whose size is at most $p$, the $k$-max-distance sparsifier of $\mathcal{D}^*_{C}$ can be constructed using the algorithm in Section~\ref{sec:overview_small}.
The exact empty extension oracle on $\mathcal{D}^*_{C}$ corresponds to the exact extension oracle on $\mathcal{D}$.

Here, we note the difference between our framework and that used by Fomin et al.~\cite{fomin2024diversecollection} and Funayama et al.~\cite{funayama2024parameterized} to provide FPT algorithms for the max-min diversification problem on $\mathcal{D}$ when $\mathcal{D}$ is the family of perfect matchings and shortest paths, respectively.
Their algorithms also start by dividing $\mathcal{D}$ into clusters.
However, their algorithms perform stricter clustering than ours.
Specifically, in their clustering, clusters $\mathcal{D}_{C}$ corresponding to different $C\in \mathcal{C}$ must be well-separated.
In contrast, we allow clusters to overlap. This simplifies the clustering step compared to their approach at the cost of a more challenging task afterward. We resolve this more challenging task by introducing and designing the $d$-limited $k$-max-distance sparsifier.

\subsection{Further Related Work}\label{app:further_related}
\subsubsection{Further Algorithms for Diversification Problems}
The max-sum diversification problem is often more tractable compared to the max-min diversification problem, and polynomial-time algorithms are known for multiple domains $\mathcal{D}$.
Hanaka et al.~\cite{hanaka2021finding} provided a polynomial-time algorithm for the case where $\mathcal{D}$ is the base family of a matroid.
Hanaka et al.~\cite{hanaka2022computing} provided polynomial-time algorithms for the cases where $\mathcal{D}$ is the family of shortest paths, branchings, and bipartite matchings.
de Berg et al.~\cite{BergMS23} provided a polynomial-time algorithm for the case where $\mathcal{D}$ is the family of edge sets of minimum $s,t$-cuts.
There has also been active research on approximation algorithms for the max-sum diversification problem.
Hanaka et al.~\cite{hanaka2023framework} proposed a generic framework that provides local search-based approximation algorithms for max-sum diversification problems.
Gao et al.~\cite{gao2022obtaining} provided a framework for bicriteria approximation algorithms for the case where $\mathcal{D}$ is a family of (not necessarily optimal) solutions of an optimization problem.
Do et al.~\cite{do2023diverse} discussed the tradeoff between solution quality and diversity for submodular maximization on matroids.

There are several other research directions on parameterized algorithms for diversification problems.
Drabik and Masa{\v{r}}{\'\i}k~\cite{drabik2024finding} provided an FPT algorithm for diversification problems on domains expressible by $\mathrm{MSO}_1$ formulas on graphs with bounded cliquewidth, parameterized by $k$, $d$, cliquewidth, and the length of the $\mathrm{MSO}_1$ formula.
Arrighi et al.~\cite{arrighi2023synchronization} proposed an FPT algorithm for the max-min diversification problem on the set of synchronizing words for a deterministic finite automaton, where the distance between two words is measured by the edit distance.
Misra et al.~\cite{misra2024parameterized}~and Gima et al.~\cite{gima2024computing} independently investigated the parameterized complexity of the diverse pair of satisfiability problems.
Merkl et al.~\cite{merkl2025diversity} investigated the parameterized complexity of diversification problems on answers for conjunctive queries on relational databases.

\subsubsection{Parameterized Algorithms for Clustering Problems}

For the case where the point set is explicitly given, several research have been conducted on parameterized algorithms for clustering problems~\cite{bandyapadhyayL023a, chen2024parameterized, demaine2005fixed,eiben2023parameterized, feldmann2020parameterized, inamdar2020capacitated}. 
Probably, the research most closely related to our situation is the extension of the \emph{closest string problem} by Amir et al.~\cite{amir2014efficiency}. 
They defined the $k$-center and $k$-sum-of-radii clustering problems for sets of strings and investigated their parameterized complexity.
When the alphabet is binary, they consider the same problem as ours. 
However, their results do not directly imply our results since, in our setting, the points are implicitly given as the solution domain of combinatorial problems, which are generally of exponential size.


\subsubsection{Core-sets and Kernelization}

In the context of machine learning and data mining, \emph{core-sets}~\cite{agarwal2004approximating,feldman2020introduction} are widely investigated to reduce the size of the input point set. 
Although the precise definition varies depending on the context, a core-set typically refers to a subset of an input point set such that the solving a specific optimization problem on the core-set approximates the optimal solution of the same problem on the original set. 
Research has been conducted on constructing core-sets for diversification and clustering problems. 
In particular, Indyk et al.~\cite{indyk2014composable} constructed a \emph{composable core-set} for the max-min diversification problem on explicitly given point sets, which have applications in streaming and distributed computing. 
Core-sets for clustering problems are a popular research topic, and have been extensively studied for clustering problems such as $k$-means~\cite{chen2009coresets,har2004coresets}, $k$-median~\cite{badoiu2002approximate,chen2009coresets,har2004coresets}, $k$-center~\cite{badoiu2002approximate}, and $k$-sum-of-radii~\cite{inamdar2020capacitated}.

Our max-distance sparsifier can be viewed as a variant of core-sets, although some points are different from typical ones. 
Firstly, max-distance sparsification requires that the optimal value for the sparsifier is equal to that for the original instance, rather than approximating it. 
Secondly, in max-distance sparsification, $\mathcal{D}$ is not given explicitly but is implicitly given as a family of solutions to a combinatorial problem. 
Thirdly, whereas many core-sets are defined for problems in Euclidean metrics, this study deals with the Hamming metric.
Due to these differences, max-distance sparsifiers can be used to obtain exact algorithms for diversification and clustering problems on combinatorial domains.

In the context of parameterized algorithms, \emph{Kernelization}~\cite{fomin2019kernelization} refers to techniques for transforming an input into a smaller instance while preserving the optimal solution. Precisely, a \emph{kernelization} is a polynomial-time algorithm that, given an instance of a parameterized problem, produces an equivalent instance of constant size. In diversification problems, a max-distance sparsifier is similar to a kernel in the sense that it is a constant-size instance preserving the existence of a solution. However, while kernelization must be performed in polynomial time, max-distance sparsification is allowed to use FPT time.

\subsubsection{Diversification in AI Fields}

Diversification has been extensively studied in many areas of artificial intelligence, such as recommendation~\cite{ziegler2005improving}, ranking~\cite{clarke2008novelty}, fairness~\cite{CelisSV18}, and voting~\cite{LuB11}, as a method to prevent homogeneity and redundancy in the output.  
Most of these cases consider diversification over an explicitly given set of elements.  
Our setting corresponds to the case where each element is a solution to a combinatorial problem. 
We believe that this assumption captures many realistic scenarios where the output elements are not predefined, but must be computed, e.g., a car navigation system generating and recommending diverse shortest paths.

\subsection{Organization}

The rest of this paper is organized as follows. 
In Section~\ref{sec:usage}, we provide FPT algorithms for solving diversification and clustering problems on $\mathcal{D}$ using a constant-size $d$-limited $k$-max-distance sparsifier of $\mathcal{D}$.
In Section~\ref{sec:framework_small}, we prove Theorem~\ref{thm:framework_small} by providing an FPT oracle algorithm parameterized by $k+\ell $ that computes a constant-size $k$-max-distance sparsifier of $\mathcal{D}$.
In Section~\ref{sec:framework_large}, we prove Theorem~\ref{thm:framework_large} by providing an FPT oracle algorithm parameterized by $k+d$ that computes a constant-size $d$-limited $k$-max-distance sparsifier of $\mathcal{D}$. 
The discussion in Section~\ref{sec:framework_large} internally uses the results from Section~\ref{sec:framework_small}.
Finally, in Section~\ref{sec:application}, we apply the results of Theorems~\ref{thm:framework_small} and~\ref{thm:framework_large} to several domains $\mathcal{D}$ to obtain FPT algorithms for diversification and clustering problems.

\section{From Sparsifier to Diversification and Clustering}\label{sec:usage}

In this section, we provide FPT algorithms for diversification and clustering problems using a $d$-limited $k$-max-distance sparsifier of constant size.

\subsection{Diversification}\label{sec:usage_diverse}
Let $U$ be a finite set, $d\in \mathbb{Z}_{\geq 0}$, $k\in \mathbb{Z}_{\geq 1}$, and $\mathcal{D}\subseteq 2^U$.
For diversification problems, we have the following.
\begin{lemma}\label{lem:diversification_replace}
Let $\mathcal{K}\subseteq \mathcal{D}$ be a $d$-limited $(k-1)$-max-distance sparsifier of $\mathcal{D}$ and $(D_1, \dots, D_k)\in \mathcal{D}^k$.
Then, there is a $k$-tuple $(K_1, \dots, K_k) \in \mathcal{K}^k$ such that $\min(d, |D_i \triangle D_j|) \leq \min(d, |K_i \triangle K_j|)$ holds for all $1 \leq i < j \leq k$.
Particularly, if there exists a solution to the max-min/max-sum diversification problem on $\mathcal{D}$, then there exists a solution consisting only of sets in $\mathcal{K}$.
\end{lemma}
\begin{proof}
Assume $(D_1, \dots, D_k)\in \mathcal{D}^k \setminus \mathcal{K}^k$ and let $i \in \{1, \dots, k\}$ be an index such that $D_i \not\in \mathcal{K}$.
It is sufficient to prove that there is a set $K_i \in \mathcal{K}$ such that $\min(d, |D_i \triangle D_j|) \leq \min(d, |K_i \triangle D_j|)$ holds for all $j \in \{1, \dots, k\} \setminus \{i\}$.
For $j \in \{1, \dots, k\} \setminus \{i\}$, let $z_j := \min(d, |D_i \triangle D_j|)$. Since $\mathcal{K}$ is a $d$-limited $(k-1)$-max-distance sparsifier of $\mathcal{D}$, there exists $K_i \in \mathcal{K}$ such that $\min(d, |K_i \triangle D_j|) \geq \min(d, z_j) = \min(d, |D_i \triangle D_j|)$ holds for all $j \in \{1, \dots, k\} \setminus \{i\}$.
\end{proof}
Considering an algorithm that exhaustively searches for a subfamily of $\mathcal{K}$ of size $k$, we can state the following.
\begin{lemma}\label{lem:diversification}
Assume there exists an FPT algorithm parameterized by $k+d$ to compute a $d$-limited $(k-1)$-max-distance sparsifier of $\mathcal{D}$ with size bounded by a constant that depends only on $k+d$. 
Then, there exists an FPT algorithm parameterized by $k+d$ for the max-min/max-sum diversification problem on $\mathcal{D}$.
\end{lemma}
We note that, under the slightly stronger assumption, the discussion in this section can directly be extended to the case where the sets $D_1, \dots, D_k$ are taken from different domains.
Specifically, let $\mathcal{D}_1,\dots, \mathcal{D}_k\subseteq 2^U$ and assume $(k-1)$-max-distance sparsifiers $\mathcal{K}_1,\dots, \mathcal{K}_k$ of these domains with respect to $2^U$ are computed.
Then, we can determine whether there exists a $k$-tuple $(D_1,\dots, D_k)\in \mathcal{D}_1\times \dots\times \mathcal{D}_k$ such that $\min_{1\leq i<j\leq k}|D_i\triangle D_j|\geq d$ (or $\sum_{1\leq i<j\leq k}|D_i\triangle D_j|\geq d$) by exhaustive search on $\mathcal{K}_1\times \dots\times \mathcal{K}_k$.

\subsection{Clustering}\label{sec:clustering}

Let $U$ be a finite set, $d\in \mathbb{Z}_{\geq 0}$, $k\in \mathbb{Z}_{\geq 1}$, and $\mathcal{D} \subseteq 2^U$.
Here, we provide an FPT algorithm for the $k$-center and $k$-sum-of-radii clustering problems using a $(d+1)$-limited $k$-max-distance sparsifier $\mathcal{K}$ of $\mathcal{D}$.
For $Z\subseteq U$ and $r\in \mathbb{Z}_{\geq 0}$, the \emph{ball} of radius $r$ centered at $Z$ is defined as $\mathcal{B}(Z,r):=\{Z'\subseteq U\colon |Z\triangle Z'|\leq r\}$.
The algorithm first guesses a partition of $\mathcal{K}$ into $k$ clusters $\mathcal{K}_1, \dots, \mathcal{K}_k$.
Since $|\mathcal{K}|$ is constant, the cost of this guess is constant.
Then, for each $i \in \{1, \dots, k\}$, the algorithm computes the minimum radius $r_i$ such that there is a set $D_i\in \mathcal{D}$ satisfying $\mathcal{K}_i\subseteq \mathcal{B}(D_i,r_i)$. If $r_i > d$, the algorithm asserts it instead of computing the specific value of $r_i$.
The $k$-center clustering problem and $k$-sum-of-radii clustering problem on $\mathcal{D}$ are solved by checking whether the maximum and sum, respectively, of the $r_i$s is at most $d$.
We show the correctness of this algorithm by proving the following.
\begin{lemma}\label{lem:clustering_kernel}
Let $(D_1, \dots, D_k) \in \mathcal{D}^k$ and $(r_1, \dots, r_k) \in \{0, \dots, d\}^k$.
Assume $\mathcal{K}_i\subseteq \mathcal{B}(D_i,r_i)$ holds for all $i \in \{1, \dots, k\}$. 
Then, for all $D \in \mathcal{D}$, there is an index $i \in \{1, \dots, k\}$ such that $D\in \mathcal{B}(D_i,r_i)$.
\end{lemma}
\begin{proof}
Assume the contrary. Then, there is a set $D \in \mathcal{D}$ such that for all $i \in \{1, \dots, k\}$, $|D_i \triangle D| \geq r_i + 1$. Since $\mathcal{K}$ is a $(d+1)$-limited $k$-max-distance sparsifier of $\mathcal{D}$, there is a set $K \in \mathcal{K}$ such that for all $i \in \{1, \dots, k\}$, $|D_i \triangle K| \geq r_i + 1$. Hence, $K \not\in \mathcal{K}_i$ for all $i \in \{1, \dots, k\}$, contradicting the fact that $(\mathcal{K}_1, \dots, \mathcal{K}_k)$ is a partition of $\mathcal{K}$.
\end{proof}

We now provide an algorithm to decide whether there exists $D \in \mathcal{D}$ with $\mathcal{K}_i\subseteq \mathcal{B}(D,r_i)$ for each $i\in \{1,\dots, k\}$ and $r_i\in \{0,\dots, d\}$.
If the domain $\mathcal{D}$ is $2^U$, this problem is equivalent to the \emph{closest string problem} on binary strings, for which a textbook FPT algorithm parameterized by $d+|\mathcal{K}_i|$ is known~\cite{cygan2015parameterized}.
Our algorithm is a modified version of this.
An element $e \in U$ is \emph{bad} if there exist both $K \in \mathcal{K}_i$ with $e \in K$ and $K \in \mathcal{K}_i$ with $e \not \in K$. The following lemma is fundamental.
\begin{lemma}[\rm{\cite{cygan2015parameterized}}]
If there are more than $d|\mathcal{K}_i|$ bad elements, no $D\in \mathcal{D}$ satisfies $\mathcal{K}_i\subseteq \mathcal{B}(D,d)$.
\end{lemma}
Let $B$ be the set of bad elements, and assume $|B|\leq d|\mathcal{K}_i|$. The algorithm first guesses $B'\subseteq B$. The cost of this guess is $2^{d|\mathcal{K}_i|}$. Then, it determines whether there exists $D \in \mathcal{D}$ such that $D\cap B=B'$ and $\mathcal{K}_i\subseteq \mathcal{B}(D,r_i)$. 
Let $K^* = \mathrm{argmax}_{K \in \mathcal{K}_i} |(K\cap B) \triangle B'|$. Then, we can claim the following.
\begin{lemma}
For $D \in \mathcal{D}$ such that $D\cap B=B'$, $\max_{K \in \mathcal{K}_i} |K \triangle D| = |K^* \triangle D|$.
\end{lemma}
\begin{proof}
Let $K \in \mathcal{K}_i$. Then, $|K \triangle D| = |(K \cap B) \triangle (D \cap B)| + |(K \setminus B) \triangle (D \setminus B)|$. From the definition of $B$, the value of $|(K \setminus B) \triangle (D \setminus B)|$ is equal among all $K \in \mathcal{K}_i$. Thus, the maximum value of $|K \triangle D|$ for $K \in \mathcal{K}_i$ is achieved by the set $K$ that maximizes $|(K \cap B) \triangle (D \cap B)| = |(K\cap B) \triangle B'|$.
\end{proof}
Now, it is sufficient to solve the problem of determining whether there exists $D \in \mathcal{D}$ such that $D \cap B = B'$ and $|K^* \triangle D| \leq r_i$. This corresponds to the exact extension oracle on $\mathcal{D}$ with $r = r_i$, $X = B'$, $Y = B \setminus B'$, and $C = K^*$.
Therefore, we can claim the following:
\begin{lemma}\label{lem:clustering}
Assume there exists an FPT algorithm parameterized by $k+d$ that computes a $(d+1)$-limited $k$-max-distance sparsifier of $\mathcal{D}$ with size bounded by a constant that depends only on $k+d$, and the exact extension oracle on $\mathcal{D}$ whose time complexity is FPT parameterized by $r + |X| + |Y|$. Then, there exists an FPT algorithm parameterized by $k+d$ for the $k$-center/$k$-sum-of-radii clustering problem on $\mathcal{D}$.
\end{lemma}

\subsection{Modified Hamming Distance}\label{sec:modified_hamming}

In some cases, such as when $\mathcal{D}$ is a edge bipartization domain, sets $Z\in \mathcal{D}$ and $U \setminus Z \in \mathcal{D}$ should be considered equivalent. 
In such cases, it is natural to define the distance between sets $D_1, D_2 \in \mathcal{D}$ using the \emph{modified Hamming distance} $|D_1 \bar{\triangle} D_2| := \min(|D_1 \triangle D_2|, |D_1 \triangle (U \setminus D_2)|)$ instead of the Hamming distance $|D_1 \triangle D_2|$.
We can still design FPT algorithms for diversification and clustering problems for this modified Hamming distance using a $d$-limited $k$-max-distance sparsifier. 
We show that results similar to Lemmas~\ref{lem:diversification}~and~\ref{lem:clustering} also hold for this distance.
We assume $(U\setminus D) \in \mathcal{D}$ holds for all $D \in \mathcal{D}$.
We prove the following lemmas.
\begin{lemma}\label{lem:diversification_sym}
Assume there exists an FPT algorithm parameterized by $k+d$ that computes a $d$-limited $(2k-2)$-max-distance sparsifier of $\mathcal{D}$ with size bounded by a constant that depends only on $k+d$.
Then, there exists an FPT algorithm parameterized by $k+d$ for the max-sum/max-min diversification problem on $\mathcal{D}$ for the modified Hamming distance.
\end{lemma}
\begin{lemma}\label{lem:clustering_sym}
Assume there exists an FPT algorithm parameterized by $k+d$ that computes a $(d+1)$-limited $2k$-max-distance sparsifier of $\mathcal{D}$ with size bounded by a constant that depends only on $k+d$, and the exact extension oracle on $\mathcal{D}$ whose time complexity is FPT parameterized by $r + |X| + |Y|$. 
Then, there exists an FPT algorithm parameterized by $k+d$ for the $k$-center/$k$-sum-of-radii clustering problem on $\mathcal{D}$ for the modified Hamming distance.
\end{lemma}

First, we prove Lemma~\ref{lem:diversification_sym}.
We prove a corresponding result to Lemma~\ref{lem:diversification_replace}. The remaining part proceeds similarly as in Lemma~\ref{lem:diversification}.
\begin{lemma}
Let $U$ be a finite set, $k \in \mathbb{Z}_{\geq 1}$, $d \in \mathbb{Z}_{\geq 0}$, and $\mathcal{D} \subseteq 2^U$.
Let $\mathcal{K}$ be a $d$-limited $(2k-2)$-max-distance sparsifier of $\mathcal{D}$ and $(D_1, \dots, D_k) \in \mathcal{D}^k$.
Then, there is a $k$-tuple $(K_1, \dots, K_k) \in \mathcal{K}^k$ such that $\min(d, |D_i \bar{\triangle} D_j|) \leq \min(d, |K_i \bar{\triangle} K_j|)$ holds for all $1 \leq i < j \leq k$.
\end{lemma}
\begin{proof}
Assume $(D_1, \dots, D_k) \in \mathcal{D}^k \setminus \mathcal{K}^k$ and let $i$ be an index such that $D_i \not \in \mathcal{K}$.
It is sufficient to prove that there is a set $K_i \in \mathcal{K}$ such that $\min(d, |D_i \bar{\triangle} D_j|) \leq \min(d, |K_i \bar{\triangle} D_j|)$ holds for all $j \in \{1, \dots, k\} \setminus \{i\}$.
If $(U \setminus D_i) \in \mathcal{K}$, $K_i := U \setminus D_i$ satisfies this condition.
Assume otherwise.
For $j \in \{1, \dots, k\} \setminus \{i\}$, let $z_j:=\min(d, |D_i \triangle D_j|)$ and $z'_j := \min(d, |D_i \triangle (U\setminus D_j)|)$.
Since $\mathcal{K}$ is a $d$-limited $(2k-2)$-max-distance sparsifier of $\mathcal{D}$, there exists $K_i \in \mathcal{K}$ such that for all $j \in \{1, \dots, k\} \setminus \{i\}$,
\begin{align*}
\min(d, |K_i \bar{\triangle} D_j|) 
&= \min\left(\min(d, |K_i \triangle D_j|), \min(d, |K_i \triangle (U \setminus D_j)|)\right) \geq \min(\min(d, z_j), \min(d, z'_j))\\
&= \min(d, |D_i \triangle D_j|, |D_i \triangle (U \setminus D_j)|) = \min(d, |D_i \bar{\triangle} D_j|). \qedhere
\end{align*}
\end{proof}

Next, we prove Lemma~\ref{lem:clustering_sym}.
Let $\mathcal{K}$ be a $(d+1)$-limited $2k$-max-distance sparsifier of $\mathcal{D}$.
Similar to Lemma~\ref{lem:clustering}, the algorithm first guesses a partition of $\mathcal{K}$ into $k$ clusters $\mathcal{K}_1, \dots, \mathcal{K}_k$.
Subsequently, for each $i \in \{1, \dots, k\}$, the algorithm determines whether the minimum radius $r_i$ of a ball (in the sense of modified Hamming distance) centered at a set of $\mathcal{D}$ that covers $\mathcal{K}_i$ is at most $d$, and if so, computes this value.
We prove a result corresponding to Lemma~\ref{lem:clustering_kernel}.
\begin{lemma}
Let $(D_1, \dots, D_k) \in \mathcal{D}^k$ and $(r_1, \dots, r_k) \in \{0, \dots, d\}^k$.
Assume $|D_i \bar{\triangle} K| \leq r_i$ holds for all $i \in \{1, \dots, k\}$ and $K \in \mathcal{K}_i$. Then, for all $D \in \mathcal{D}$, there is an index $i \in \{1, \dots, k\}$ such that $|D_i \bar{\triangle} D| \leq r_i$.
\end{lemma}
\begin{proof}
Assume the contrary. Then, there is a set $D \in \mathcal{D}$ such that for all $i \in \{1, \dots, k\}$, $|D \bar{\triangle} D_i| \geq r_i + 1$, which means $|D \triangle D_i| \geq r_i + 1$ and $|D \triangle (U \setminus D_i)| \geq r_i + 1$. Since $\mathcal{K}$ is a $(d+1)$-limited $2k$-max-distance sparsifier of $\mathcal{D}$, there is a set $K \in \mathcal{K}$ such that for all $i \in \{1, \dots, k\}$, $|K \bar{\triangle} D_i| = \min(|K \triangle D_i|, |K \triangle (U \setminus D_i)|) \geq r_i + 1$. Hence, $K \not \in \mathcal{K}_i$ for all $i \in \{1, \dots, k\}$, contradicting the fact that $(\mathcal{K}_1, \dots, \mathcal{K}_k)$ is a partition of $\mathcal{K}$.
\end{proof}

We now provide an algorithm to evaluate the minimum radius of the ball containing $\mathcal{K}_i$ in terms of the modified Hamming distance.
To determining whether there is $D_i \in \mathcal{D}$ such that $|D_i \bar{\triangle} K| \leq r_i$ for all $K \in \mathcal{K}_i$, we guess $\bar{K} \in \{K, U \setminus K\}$ for each $K \in \mathcal{K}_i$ and determine whether there is $D_i\in \mathcal{D}$ such that $|D_i\triangle \bar{K}|\leq r_i$ for all $K\in \mathcal{K}_i$. After this guessing, the rest are the same as the case discussed in Section~\ref{sec:clustering}.

\section{Framework for $k$-Max-Distance Sparification}\label{sec:framework_small}

In this section, we complete the proof of Theorem~\ref{thm:framework_small} by providing an FPT algorithm that uses the exact empty extension oracle on $\mathcal{D}$ to obtain a $k$-max-distance sparsifier of $\mathcal{D}$. 
For further use, we show a slightly more extended result. 
Let $r \in \mathbb{Z}_{\geq 0}$. 
We construct a $k$-max-distance sparsifier of $\mathcal{D}$ with respect to $\mathcal{B}(\emptyset, r)$ for $r\geq \ell $. 
Since $\mathcal{D} \subseteq \mathcal{B}(\emptyset, \ell) \subseteq \mathcal{B}(\emptyset, r)$ for $r \geq \ell $, this is also a $k$-max-distance sparsifier of $\mathcal{D}$ (with respect to $\mathcal{D}$).
A set family $\mathcal{S}:=\{S_1, \dots, S_t\}$ is called a \emph{sunflower} if there exists a set called \emph{core} $C$ such that for any $1 \leq i < j \leq t$, $S_i \cap S_j = C$.
The following is well-known.
\begin{lemma}[\rm{Sunflower Lemma~\cite{cygan2015parameterized, erdos1960intersection}}]\label{lem:sunflower}
Let $U$ be a finite set, $\ell, t \in \mathbb{Z}_{\geq 0}$, and $\mathcal{K} \subseteq 2^U$ be a family consisting only of sets of size at most $\ell $. If $|\mathcal{K}| > \ell ! (t - 1)^\ell $, then $\mathcal{K}$ contains a sunflower of size $t$.
\end{lemma}
For $t\in \mathbb{Z}_{\geq 0}$, $\mathcal{T}\subseteq 2^U$, and $Z\in 2^U \setminus \mathcal{T}$, a sunflower $\mathcal{S}\subseteq \mathcal{T}$ is a \emph{$(Z,t)$-sunflower of $\mathcal{T}$} if it satisfies the following three conditions.
\begin{itemize}
    \item $|\mathcal{S}|=t$,
    \item For each $S\in \mathcal{S}$, $|S|=|Z|$, and
    \item The core of $\mathcal{S}$ is a subset of $Z$.
\end{itemize}
The following lemma is the core of our framework.
\begin{lemma}\label{lem:previouspaper}
Let $U$ be a finite set, $\mathcal{D}\subseteq 2^U$, and $\mathcal{K}\subseteq \mathcal{D}$ be a $k$-max-distance sparsifier of $\mathcal{D}$ with respect to $\mathcal{B}(\emptyset, r)$.
Let $Z\in \mathcal{K}$ and assume there is a $(Z, kr+1)$-sunflower $\mathcal{S}$ of $\mathcal{K} \setminus \{Z\}$.
Then, $\mathcal{K}\setminus \{Z\}$ is also a $k$-max-distance sparsifier of $\mathcal{D}$ with respect to $\mathcal{B}(\emptyset, r)$.
\end{lemma}
\begin{proof}
Let $(F_1, \dots, F_k) \in \mathcal{B}(\emptyset, r)^k$ and $(z_1, \dots, z_k) \in \mathbb{Z}_{\geq 0}^k$. We show the equivalence of the following two conditions:
\begin{itemize}
    \item There exists $K \in \mathcal{K}$ such that for each $i \in \{1, \dots, k\}$, $|F_i \triangle K| \geq z_i$.
    \item There exists $K \in \mathcal{K} \setminus \{Z\}$ such that for each $i \in \{1, \dots, k\}$, $|F_i \triangle K| \geq z_i$.
\end{itemize}
Since $\mathcal{K} \setminus \{Z\} \subseteq \mathcal{K}$, the latter implies the former. We assume the former and prove the latter.
Take the $K$ that satisfies the former condition. If $K\neq Z$, the claim is obvious, so assume $K=Z$. Then, there exists a $(Z,kr+1)$-sunflower $\mathcal{S}$ of $\mathcal{K}\setminus \{Z\}$. Let $C$ be the core of $\mathcal{S}$.
Since $\left|\bigcup_{i\in \{1,\dots, k\}}F_i\right|\leq \sum_{i\in \{1,\dots, k\}}|F_i|\leq kr$, there exists an $S\in \mathcal{S}$ such that $F_i \cap (S\setminus C) = \emptyset$ for all $i\in \{1, \dots, k\}$.
In this case, for each $i \in \{1, \dots, k\}$, $|F_i \triangle S| = |F_i| + |S| - 2|F_i \cap S| = |F_i| + |Z| - 2|F_i \cap C| \geq |F_i| + |Z| - 2|F_i \cap Z| = |F_i \triangle Z| \geq z_i$, where the first inequality follows from $C \subseteq Z$.
Thus, the lemma is proved.
\end{proof}

\begin{algorithm}[t!]
\caption{$k$-max-distance sparsification of $\mathcal{D}$.}\label{alg:k-sparsify}
\Procedure{\emph{\Call{KSparsify}{$k,r$}}}{
    \KwIn{$k\in \mathbb{Z}_{\geq 1}$, $r \in \mathbb{Z}_{\geq 0}$}
    Let $\mathcal{K} := \emptyset$\;
    \While{$\mathsf{true}$}{\label{line:k-sparsify-loop}
        Let $R:= \bigcup_{K\in \mathcal{K}}K$, $f := \mathsf{false}$\;
        \For{$\ell'\in \{0,\dots, l\}$}{
            Let $\mathfrak{S}$ be the family of all sunflowers $\mathcal{S} \subseteq \mathcal{K}$ such that $|\mathcal{S}| = kr + 1$ and each $S \in \mathcal{S}$ satisfies $|S| = l'$\;
            \For{$Y \subseteq R$ that intersects with all $K \in \mathcal{K}$ with $|K| = l'$ and the cores of all sunflowers of $\mathfrak{S}$}{\label{line:setting_y}
                Let $D = \Call{ExactEmptyExtension}{l', Y}$\;
                \If{$D \neq \bot$ and $f = \mathsf{false}$}{
                    Add $D$ to $\mathcal{K}$ and $f := \mathsf{true}$\;\label{line:exemex}
                }
            }
        }
        \If{$f = \mathsf{false}$}{
            \bf{break}\;
        }
    }
    \Return $\mathcal{K}$\;
}
\end{algorithm}

Our algorithm is given in Algorithm~\ref{alg:k-sparsify}, where \Call{ExactEmptyExtension}{$\ell', Y$} represents the exact empty extension oracle on $\mathcal{D}$ with arguments $\ell'$ and $Y$.
The algorithm starts with $\mathcal{K} := \emptyset$ and repeatedly adds $Z \in \mathcal{D} \setminus \mathcal{K}$ such that there is no $(Z, kr+1)$-sunflower of $\mathcal{K}$ to $\mathcal{K}$.
The following lemma shows this algorithm stops after a constant number of iterations.

\begin{lemma}\label{lem:numloops}
The number of iterations of the loop starting from line~\ref{line:k-sparsify-loop} in Algorithm~\ref{alg:k-sparsify}, as well as the size of the output family, is at most $(\ell+1)!(kr+1)^\ell $.
\end{lemma}
\begin{proof}
By the definition of $Y$ in line~\ref{line:setting_y}, the $D$ added to $\mathcal{K}$ in line~\ref{line:exemex} is distinct from any set in $\mathcal{K}$, since $Y$ has a nonempty intersection with each $D \in \mathcal{K}$ with $|D|=\ell'$. Moreover, there is no $(D,kr+1)$-sunflower in $\mathcal{K}$, since $Y$ has a nonempty intersection with the core of any sunflower in $\mathfrak{S}$.
Thus, adding $D$ to $\mathcal{K}$ does not form a new sunflower of size $kr+2$ consisting of sets of size $|D|$. 
Hence, by Lemma~\ref{lem:sunflower}, for each $\ell'\in \{0,\dots, \ell\}$, the number of sets of size $\ell'$ in $\mathcal{K}$ is at most $\ell !(kr+1)^\ell $. Thus, the number of iterations of the loop and the size of the output family is at most $(\ell+1)!(kr+1)^\ell $.
\end{proof}

In particular, at each step of the algorithm, since $|R| \leq |\mathcal{K}|l \leq (\ell+1)!(kr+1)^\ell \ell $, the size of $Y$ chosen in line~\ref{line:setting_y} is bounded by a constant.
The time complexity is bounded as follows.
\begin{lemma}\label{lem:numextension}
Algorithm~\ref{alg:k-sparsify} makes at most $2^{2^{O(\ell \log (klr))}}$ calls of \Call{ExactEmptyExtension}{$\cdot$} and has a time complexity of $2^{2^{O(\ell \log (klr))}}$.
\end{lemma}
\begin{proof}
By Lemma~\ref{lem:numloops}, the number of iterations of the loop starting from line~\ref{line:k-sparsify-loop} is $2^{O(\ell \log (k\ell r))}$. Therefore, at each step of the algorithm, $|R|\leq 2^{O(\ell\log (k\ell r))}$.
We can compute $\mathfrak{S}$ by exhaustively checking subsets of $\mathcal{K}$ of size $kr+1$, which takes $|\mathcal{K}|^{kr+1} \leq 2^{O(klr\log(k\ell r))}$ time.
The number of candidates for $Y$ in line~\ref{line:setting_y} is at most $2^{|R|} \leq 2^{2^{O(\ell \log k\ell r)}}$, which dominates the time complexity.
\end{proof}

The correctness of the algorithm is shown as follows.
\begin{lemma}\label{lem:ksparsify_correctness}
Algorithm~\ref{alg:k-sparsify} outputs a $k$-max-distance sparsifier of $\mathcal{D}$ with respect to $\mathcal{B}(\emptyset, r)$.
\end{lemma}
\begin{proof}
Let $\mathcal{K}$ be the output of Algorithm~\ref{alg:k-sparsify} and $\mathcal{D} \setminus \mathcal{K} = \{D_1, \dots, D_{|\mathcal{D}\setminus \mathcal{K}|}\}$.
By the termination condition of the algorithm, for each $i \in \{1, \dots, |\mathcal{D}\setminus \mathcal{K}|\}$, there exists a $(D_i, kr+1)$-sunflower of $\mathcal{K}$.
By Lemma~\ref{lem:previouspaper}, for each $i \in \{1, \dots, |\mathcal{D}\setminus \mathcal{K}|\}$, if $\mathcal{D} \setminus \{D_1, \dots, D_{i-1}\}$ is a $k$-max-distance sparsifier of $\mathcal{D}$ with respect to $\mathcal{B}(\emptyset, r)$, then $\mathcal{D} \setminus \{D_1, \dots, D_i\}$ also is. Therefore, $\mathcal{K} = \mathcal{D} \setminus \{D_1, \dots, D_{|\mathcal{D}\setminus \mathcal{K}|}\}$ is a $k$-max-distance sparsifier of $\mathcal{D}$ with respect to $\mathcal{B}(\emptyset, r)$.
\end{proof}
\section{Framework for $d$-Limited $k$-Max-Distance Sparsification}\label{sec:framework_large}

\subsection{Overall Flow}\label{sec:overall}

In this section, we complete the proof of Theorem~\ref{thm:framework_large} by providing FPT algorithm that uses the $(-1,1)$-optimization oracle and the exact extension oracle on $\mathcal{D}$ to obtain a $d$-limited $k$-max-distance sparsifier of $\mathcal{D}$.
Actually, for further applications, we construct the slightly more general object of $d$-limited $k$-max-distance sparsifier of $\mathcal{D}$ with respect to $2^U$, not with respect to $\mathcal{D}$ itself.
Our framework consists of two steps.
Let $p \in \mathbb{Z}_{\geq 0}$ be an integer with $2d < p$.
The first step achieves one of the following.
\begin{itemize}
    \item Find a set $\mathcal{C}\subseteq \mathcal{D}$ with size at most $k$ such that $\mathcal{D}\subseteq \bigcup_{C\in \mathcal{C}}\mathcal{B}(C,p)$.
    \item Find a set $\mathcal{C}\subseteq \mathcal{D}$ of size $k+1$ such that $|C\triangle C'| >2d$ holds for any distinct $C,C'\in \mathcal{C}$.
\end{itemize}
We do this by using the following \emph{approximate far set oracle}, which will be designed in Section~\ref{sec:findfarset}.

\begin{quote}
\textbf{Approximate Far Set Oracle:}
Let $U$ be a finite set, $d \in \mathbb{Z}_{\geq 0}$, $\mathcal{D} \subseteq 2^U$, and $\mathcal{C} \subseteq \mathcal{D}$.
The approximate far set oracle returns one of the following.
\begin{itemize}
    \item A set of $\mathcal{D}$ that does not belong to $\bigcup_{C \in \mathcal{C}}\mathcal{B}(C, 2d)$.
    \item $\bot$. This option can be chosen only when $\mathcal{D} \subseteq \bigcup_{C \in \mathcal{C}}\mathcal{B}(C, p)$.
\end{itemize}
\end{quote}

Starting with $\mathcal{C} := \emptyset$, we repeat the following steps.
If the approximate far set oracle returns $\bot$, terminate the loop.
Otherwise, add the element found by the oracle to $\mathcal{C}$. If the oracle returns $\bot$ within $k$ iterations, the first condition is achieved. If not, the set $\mathcal{C}$ after $k+1$ iterations satisfies the second condition.
In the latter case, the following lemma shows that $\mathcal{C}$ is a $d$-limited $k$-max-distance sparsifier.
\begin{lemma}\label{lem:tooscattered}
Let $r \in \mathbb{Z}_{\geq 0}$.
Let $\mathcal{C}$ be a subset of $\mathcal{D}$ of size $k+1$ such that for any distinct $C, C' \in \mathcal{C}$, $|C \triangle C'| \geq 2d$.
Then, $\mathcal{C}$ is a $d$-limited $k$-max-distance sparsifier of $\mathcal{D}$ with respect to $2^U$.
\end{lemma}
\begin{proof}
For each $F \subseteq U$, there is at most one $C \in \mathcal{C}$ such that $|F \triangle C| < d$.
Hence, for any $k$-tuple $(F_1, \dots, F_k)$ of the subsets of $U$, there exists a $C \in \mathcal{C}$ such that $|F_i \triangle C| \geq d$ for all $i \in \{1, \dots, k\}$.
In particular, for all $(z_1, \dots, z_k) \in \{0, \dots, d\}^k$, the following two conditions are always true and thus equivalent.
\begin{itemize}
    \item There exists $D \in \mathcal{D}$ such that for each $i \in \{1, \dots, k\}$, $|F_i \triangle D| \geq z_i$.
    \item There exists $C \in \mathcal{C}$ such that for each $i \in \{1, \dots, k\}$, $|F_i \triangle C| \geq z_i$.\qedhere
\end{itemize}
\end{proof}

Now, we assume the first condition. Let $\mathcal{C}$ be a subset of $\mathcal{D}$ of size at most $k$ such that $\mathcal{D}\subseteq \bigcup_{C\in \mathcal{C}}\mathcal{B}(C,p)$.
The second step involves constructing a $d$-limited $k$-max-distance sparsifier of $\mathcal{D}_C := \mathcal{D} \cap \mathcal{B}(C, p)$ with respect to $\mathcal{B}(C, p+d)$ for each $C \in \mathcal{C}$.
We prove that the union of all such $d$-limited $k$-max-distance sparsifiers obtained in this manner is a $d$-limited $k$-max-distance sparsifier of $\mathcal{D}$ with respect to $2^U$.

\begin{lemma}\label{lem:union_distance}
Assume $\mathcal{D}=\bigcup_{C\in \mathcal{C}}\mathcal{D}_C$.
For each $C \in \mathcal{C}$, let $\mathcal{K}_C \subseteq \mathcal{D}_C$ be a $d$-limited $k$-max-distance sparsifier of $\mathcal{D}_C$ with respect to $\mathcal{B}(C, p+d)$.
Then, $\mathcal{K}:= \bigcup_{C\in \mathcal{C}} \mathcal{K}_C$ is a $d$-limited $k$-max-distance sparsifier of $\mathcal{D}$ with respect to $2^U$.
\end{lemma}
\begin{proof}
Let $(F_1, \dots, F_k) \in \left(2^U\right)^k$ and $(z_1, \dots, z_k) \in \{0, \dots, d\}^k$.
We show the equivalence of the following two conditions.
\begin{itemize}
    \item There exists $D \in \mathcal{D}$ such that for each $i \in \{1, \dots, k\}$, $|F_i \triangle D| \geq z_i$.
    \item There exists $K \in \mathcal{K}$ such that for each $i \in \{1, \dots, k\}$, $|F_i \triangle K| \geq z_i$.
\end{itemize}
Since $\mathcal{K} \subseteq \mathcal{D}$, the latter implies the former. We assume the former and prove the latter.
Take any $D \in \mathcal{D}$ that satisfies the former condition.
Since $\mathcal{D} = \bigcup_{C \in \mathcal{C}} \mathcal{D}_C$, we can take $C \in \mathcal{C}$ such that $D \in \mathcal{D}_C$.
Now, define $(F'_1, \dots, F'_k) \in \mathcal{B}(C,p+d)^k$ and $(z'_1, \dots, z'_k) \in \{0, \dots, p\}^k$ as follows.
\begin{align*}
    F'_i = \begin{cases}
        F_i & \text{if } F_i \in \mathcal{B}(C, p+d) \\
        C & \text{otherwise}
    \end{cases},\quad \quad
    z'_i = \begin{cases}
        z_i & \text{if } F_i \in \mathcal{B}(C, p+d) \\
        0 & \text{otherwise}
    \end{cases}
\end{align*}
Clearly, for each $i \in \{1, \dots, k\}$, $F'_i \in \mathcal{B}(C, p+d)$.
Furthermore, if $F_i \in \mathcal{B}(C, p+d)$, then $|F'_i \triangle D| = |F_i \triangle D| \geq z_i = z'_i$. If $F_i\not \in \mathcal{B}(C,p+d)$, $|F'_i \triangle D| \geq 0 = z'_i$. Therefore, for each $i \in \{1, \dots, k\}$, $|F'_i \triangle D| \geq z'_i$.
Since $\mathcal{K}_C$ is a $d$-limited $k$-max-distance sparsifier of $\mathcal{D}_C$ with respect to $\mathcal{B}(C, p+d)$, there exists $K \in \mathcal{K}_C \subseteq \mathcal{K}$ such that for each $i \in \{1, \dots, k\}$, $|F'_i \triangle K| \geq z'_i$.
We show that this $K$ satisfies the latter condition. For each $i \in \{1, \dots, k\}$, we show that $|F_i \triangle K| \geq z_i$.
If $F_i \in \mathcal{B}(C, p+d)$, then $|F_i \triangle K| = |F'_i \triangle K| \geq z'_i = z_i$.
If $F_i\not \in \mathcal{B}(C,p+d)$, since $K \in \mathcal{K}_C \subseteq \mathcal{D}_C \subseteq \mathcal{B}(C, p)$, we have $|F_i \triangle K| \geq |F_i \triangle C| - |K \triangle C| > (p+d) - p \geq d \geq z_i$, where the first inequality is from the triangle inequality.
\end{proof}

Next, we reduce the computation of $d$-limited $k$-max-distance sparsifiers to the computation of $k$-max-distance sparsifiers of families consisting of constant-size sets, which was discussed in Section~\ref{sec:framework_small}.
For $C \in \mathcal{C}$, let $\mathcal{D}_C^* := \{D \triangle C \mid D \in \mathcal{D}_C\}$.
By the definition of $\mathcal{D}_C$, we have $\mathcal{D}_C^*\subseteq \mathcal{B}(\emptyset, p)$.
The following holds.
\begin{lemma}\label{lem:sym_union}
Let $C \in \mathcal{C}$. A subset $\mathcal{K}_C \subseteq \mathcal{D}_C$ is a $d$-limited $k$-max-distance sparsifier of $\mathcal{D}_C$ with respect to $\mathcal{B}(C,p+d)$ if and only if $\mathcal{K}_C^* := \{K \triangle C \mid K \in \mathcal{K}_C\}$ is a $d$-limited $k$-max-distance sparsifier of $\mathcal{D}_C^*$ with respect to $\mathcal{B}(\emptyset, p+d)$.
\end{lemma}
\begin{proof}
Let $(F_1, \dots, F_k) \in \mathcal{B}(C, p+d)^k$ and $(z_1, \dots, z_k) \in \{0, \dots, d\}^k$. 
For each $i \in \{1, \dots, k\}$, let $F_i^* := F_i \triangle C$. It follows that $(F_1^*, \dots, F_k^*) \in \mathcal{B}(\emptyset, p+d)^k$.
Since $|F_i \triangle D| = |(F_i \triangle C) \triangle (D \triangle C)| = |F_i^* \triangle (D \triangle C)|$, the following two conditions are equivalent.
\begin{itemize}
    \item[\rm{(i)}] There exists $D \in \mathcal{D}_C$ such that for each $i \in \{1, \dots, k\}$, $|F_i \triangle D| \geq z_i$.
    \item[\rm{(i')}] There exists $D^* \in \mathcal{D}_C^*$ such that for each $i \in \{1, \dots, k\}$, $|F_i^* \triangle D^*| \geq z_i$.
\end{itemize}
Similarly, the following two conditions are also equivalent.
\begin{itemize}
    \item[\rm{(ii)}] There exists $K \in \mathcal{K}_C$ such that for each $i \in \{1, \dots, k\}$, $|F_i \triangle K| \geq z_i$.
    \item[\rm{(ii')}] There exists $K^* \in \mathcal{K}_C^*$ such that for each $i \in \{1, \dots, k\}$, $|F_i^* \triangle K^*| \geq z_i$.
    \end{itemize}
Hence, (i) and (ii) are equivalent if and only if (i') and (ii') are equivalent, and the lemma is proved.
\end{proof}

If $\mathcal{K}_C^*$ is a $k$-max-distance sparsifier of $\mathcal{D}_C^*$ with respect to $\mathcal{B}(\emptyset, p+d)$, then it is also a $d$-limited $k$-max-distance sparsifier of $\mathcal{D}_C^*$ with respect to $\mathcal{B}(\emptyset, p+d)$ for any $d \in \mathbb{Z}_{\geq 0}$.
Therefore, a $d$-limited $k$-max-distance sparsifier $\mathcal{K}$ of $\mathcal{D}$ with respect to $\mathcal{B}(C, p+d)$ can be computed as 
\begin{align*}
    \mathcal{K} := \bigcup_{C\in \mathcal{C}}\{K^* \triangle C\colon K^* \in \mathcal{K}^*_{C}\}.
\end{align*}
From the discussion in Section~\ref{sec:framework_small}, $\mathcal{K}_C^*$ can be obtained by calling the exact empty extension oracle on $\mathcal{D}_C^*$ a constant number of times that depends on $k$, $\ell=p$, and $r=p+d$.
The exact empty extension oracle for $\mathcal{D}_C^*$ is equivalent to the exact extension oracle on $\mathcal{D}_C$ when the inputs $C, X, Y$ are taken to be $C, Y \cap C, Y \setminus C$, respectively. Therefore, $\mathcal{K}_C^*$ can be obtained by calling the exact extension oracle on $\mathcal{D}_C$ a constant number of times that depends only on $k$ and $p$.

\subsection{Designing the Approximate Far Set Oracle}\label{sec:findfarset}

Here, we design a randomized algorithm parameterized by $|\mathcal{C}|$ and $d$ for the approximate far set oracle. 
Our algorithm repeats the following sufficient number of times: It selects a weight vector $w \in \{-1, 1\}^U$ uniformly at random and finds a set $D \in \mathcal{D}$ that maximizes $w(D):=\sum_{e\in D}w_e$.
If the found $D$ does not belong to $\bigcup_{C \in \mathcal{C}} \mathcal{B}(C, 2d)$, it outputs this $D$ and terminates. If no such $D$ is found after a sufficient number of iterations, it returns $\bot$.
We now prove the correctness of the algorithm.
We can claim the following.
\begin{lemma}\label{lem:toomanyclusters}
Assume $\max_{D \in \mathcal{D}} w(D) > \max_{C \in \mathcal{C}} w(C) + 2d$. Then, the $D$ that attains the maximum on the left-hand side does not belong to $\bigcup_{C \in \mathcal{C}} \mathcal{B}(C, 2d)$.
\end{lemma}
\begin{proof}
For each $C \in \mathcal{C}$, $|D \triangle C| = \sum_{e \in D \triangle C} |w_e| \geq w(D) - w(C) > 2d$.
\end{proof}

The following lemma is the core of the analysis.
\begin{lemma}\label{lem:new_prob}
Assume $p \geq (4d+2)^2 \cdot 2^{k-1}$ and $|\mathcal{C}|\leq k$.
Let $D \in \mathcal{D}$ and assume $D \not\in \bigcup_{C \in \mathcal{C}} \mathcal{B}(C, p)$.
Then,
\begin{align*}
    \Pr\left[w(D) > \max_{C \in \mathcal{C}} w(C) + 2d\right] \geq 2^{-2^{O(k)}}.
\end{align*}
\end{lemma}
\begin{proof}
Let $\mathcal{C}^* := \{C \triangle D \colon C \in \mathcal{C}\}$ and $w^* \in \{-1, 1\}^U$ be a weight vector such that $w^*_e = w_e$ for $e \in D$ and $w^*_e = -w_e$ otherwise. 
Then, for all $C \in \mathcal{C}$, $w(C) = w(D) - w(D \setminus C) + w(C \setminus D) = w(D) - w^*(C \triangle D)$. 
Therefore, it is sufficient to prove that 
\begin{align*}
    \Pr\left[\min_{C^* \in \mathcal{C}^*} w^*(C^*) > 2d\right] \geq 2^{-2^{O(k)}}.
\end{align*}
Since $D \not \in \bigcup_{C \in \mathcal{C}} \mathcal{B}(C, p)$, $|C^*| > p$ holds for all $C^* \in \mathcal{C}^*$. Moreover, $w^*$ follows a uniform distribution over $\{-1, 1\}^U$.

Two elements $e, e' \in U$ are \emph{equivalent under $\mathcal{C}^*$} if $e \in C^*$ and $e' \in C^*$ are equivalent for all $C^* \in \mathcal{C}^*$.
We partition the ground set $U$ into $2^{|\mathcal{C}^*|}$ subsets such that elements in the same subset are equivalent under $\mathcal{C}^*$.
Specifically, for each $\mathcal{X} \subseteq \mathcal{C}^*$, define $U_{\mathcal{X}} := \{e \in U \colon \{C^* \in \mathcal{C}^* \colon C^* \ni e\} = \mathcal{X}\}$.
For each $C^* \in \mathcal{C}^*$, define $\mathcal{X}_{C^*} \subseteq \mathcal{C}^*$ as the family where $|U_{\mathcal{X}}|$ is maximized among all subsets $\mathcal{X} \subseteq \mathcal{C}^*$ with $C^* \in \mathcal{X}$.
Each $C^* \in \mathcal{C}^*$ satisfies $|C^*| > p$, and there are $2^{|\mathcal{C}^*| - 1}$ subfamilies $\mathcal{X}\subseteq \mathcal{C}^*$ with $C^* \in \mathcal{X}$, thus by the pigeonhole principle, $|U_{\mathcal{X}}| > \frac{p}{2^{|\mathcal{C}^*| - 1}} = \frac{p}{2^{|\mathcal{C}| - 1}} \geq (4d+2)^2$.

Now, assume the weight vector $w^*$ satisfies the following two conditions.
\begin{itemize}
    \item For $\mathcal{X}\subseteq \mathcal{C}^*$ such that $\mathcal{X} = \mathcal{X}_{C^*}$ holds for some $C^*\in \mathcal{C}^*$, $\left|\{e \in U_{\mathcal{X}} \colon w^*_e = 1\}\right| > \frac{|U_{\mathcal{X}}|}{2} + d$.
    \item For all other $\mathcal{X}\subseteq \mathcal{C}^*$, $\left|\{e \in U_{\mathcal{X}} \colon w^*_e = 1\}\right| \geq \frac{|U_{\mathcal{X}}|}{2}$.
\end{itemize}
Then, for each $C^* \in \mathcal{C}^*$,
\begin{align*}
    w^*(C^*) = \sum_{\mathcal{X} \ni C^*} \left(\left|\{e \in U_{\mathcal{X}} \colon w^*_e = 1\}\right| - \left|\{e \in U_{\mathcal{X}} \colon w^*_e = -1\}\right|\right) 
    > 0 \cdot (2^{|\mathcal{C}^*| - 1} - 1) + 2d \cdot 1 = 2d.
\end{align*}
Thus, we need to show that the probability that $w^*$ satisfies these two conditions is at least $2^{-2^{O(k)}}$. Each of the conditions for $\mathcal{X}\subseteq \mathcal{C}^*$ is independent, so we evaluate the probability for a fixed $\mathcal{X}$. The probability that a fixed $\mathcal{X}$ satisfies the second condition is clearly at least $\frac{1}{2}$.
Assume there exists $C^*\in \mathcal{C}^*$ with $\mathcal{X} = \mathcal{X}_{C^*}$. Let $N := |U_{\mathcal{X}}| > (4d+2)^2$.
Then,
\begin{align*}
    \Pr\left[\left|\{e \in U_{\mathcal{X}} \colon w^*_e = 1\}\right| > \frac{N}{2} + d\right] &= \frac{1}{2} \cdot \left(1 - \frac{\sum_{p = \ceil{\frac{N}{2} - d}}^{\floor{\frac{N}{2} + d}} \binom{N}{p}}{2^N}\right)\\
    &\geq \frac{1}{2} \cdot \left(1 - \frac{(2d+1) \cdot \frac{2^N}{\sqrt{N}}}{2^N}\right)
    \geq \frac{1}{2} \cdot \left(1 - \frac{2d+1}{4d+2}\right) = \frac{1}{4},
\end{align*}
where the first inequality follows from $\binom{N}{p}\leq \frac{2^N}{\sqrt{N}}$.
Thus, the probability that $w^*$ satisfies the two conditions is at least $2^{-2^{|\mathcal{C}^*|}} \cdot 4^{-|\mathcal{C}^*|} \geq 2^{-2^{O(k)}}$.
\end{proof}
By repeating the sampling of $w$ a sufficient number of times, we can state the following.
\begin{lemma}\label{lem:farsetoracle}
Let $\epsilon > 0$, $\mathcal{D}, \mathcal{C} \subseteq 2^U$, $k \in \mathbb{Z}_{\geq 1}$, and $d \in \mathbb{Z}_{\geq 0}$. Assume $|\mathcal{C}| \leq k$.
Then, there exists a randomized algorithm that runs in time $2^{2^{O(k)}} \log \epsilon^{-1}$ and satisfies the following.
\begin{itemize}
    \item If there exists $D \in \mathcal{D}$ such that $D \not\in \bigcup_{C \in \mathcal{C}} \mathcal{B}(C, (4d+2)^2 \cdot 2^k)$, the algorithm returns a set $D' \in \mathcal{D}$ satisfying $D' \not\in \bigcup_{C \in \mathcal{C}} \mathcal{B}(C, 2d)$ with probability at least $1 - \epsilon$.
    \item If not, the algorithm returns either $\bot$ or a set $D' \in \mathcal{D}$ satisfying $D' \not\in \bigcup_{C \in \mathcal{C}} \mathcal{B}(C, 2d)$.
\end{itemize}
\end{lemma}

Combining Lemma~\ref{lem:farsetoracle} with the results from Sections~\ref{sec:framework_small} and~\ref{sec:overall}, we have the following.
\begin{lemma}\label{lem:framework_analysis}
Let $\epsilon > 0$, $\mathcal{D} \subseteq 2^U$, $k \in \mathbb{Z}_{\geq 1}$, and $d \in \mathbb{Z}_{\geq 0}$.
Then, there exists a randomized algorithm that, with probability $1 - \epsilon$, computes a $d$-limited $k$-max-distance sparsifier of $\mathcal{D}$ with respect to $2^U$ with size at most $2^{2^{O(k + \log d)}}$ in time $\left(2^{2^{2^{O(k + \log d)}}}+2^{2^{O(k)}}\log \epsilon^{-1}\right)poly(|U|)$.
It uses at most $2^{2^{O(k)}} \log \epsilon^{-1}$ calls to the $(-1, 1)$-optimization oracle on $\mathcal{D}$ and at most $2^{2^{2^{O(k + \log d)}}}$ calls to the exact extension oracle on $\mathcal{D}$ such that $r\leq (4d+2)^2 2^{k-1}$ and $|X|, |Y|\leq 2^{2^{O(k + \log d)}}$.
\end{lemma}
\begin{proof}
Since the algorithm calls the approximate far set oracle at most $k+1$ times, applying Lemma~\ref{lem:farsetoracle} with an error $\frac{\epsilon}{k+1}$ claims that the algorithm computes $\mathcal{C}$ correctly with probability $1 - \epsilon$ by making $2^{2^{O(k)}} \log \epsilon^{-1}$ calls of the $(-1, 1)$-optimization oracle.
If $\mathcal{C}$ is correctly computed, applying Algorithm~\ref{alg:k-sparsify} on $\mathcal{D}^*_{C}$ for each $C \in \mathcal{C}$ computes the $d$-limited $k$-max-distance sparsifier of $\mathcal{D}$ with respect to $2^U$. 
Applying Lemma~\ref{lem:numloops} for $\ell=p$ and $r=p+d$ bounds the size of the output sparsifier by $(p+1)!(k(p+d)+1)^p\cdot k\leq 2^{O(p\log (kp(p+d)))}\leq 2^{O(2^k d^2 \log(4^d d^4 k))}\leq 2^{2^{O(k+\log d)}}$.
Applying Lemma~\ref{lem:numextension} for $\ell=p$ and $r=p+d$ bounds the number of calls to the exact extension oracle by $2^{2^{O(p \log (kp(p+d)))}} \leq 2^{2^{O(2^k d^2 \log (4^d d^4 k))}}\leq 2^{2^{2^{O(k + \log d)}}}$.
Moreover, for each call of the exact extension oracle, we have $r\leq p = (4d+2)^2 2^{k-1}$ and $|X|,|Y|\leq 2^{2^{O(k+\log d)}}\cdot p\leq 2^{2^{O(k+\log d)}}$.
\end{proof}

\section{Application}\label{sec:application}

In this section, we show that Theorems~\ref{thm:framework_small}~and~\ref{thm:framework_large} actually yield algorithms for diversification and clustering problems on several domains by constructing the required oracles for these domains.

We remark on a simple but useful fact. Assume a domain $\mathcal{D}$ can be written as the union of domains $\mathcal{D}_1, \dots, \mathcal{D}_t$. 
Then, all oracles required by Theorems~\ref{thm:framework_small}~or~\ref{thm:framework_large} for the domain $\mathcal{D}$ can be implemented by calling a corresponding oracle for each domain $\mathcal{D}_1, \dots, \mathcal{D}_t$. 
Specifically, the $(-1,1)$-optimization oracle on $\mathcal{D}$ can be realized by returning the maximum of the outputs among each call, and the exact (empty) extension oracle on $\mathcal{D}$ can be realized by returning any solution returned by a call if any call returns a non-$\bot$ solution. 
Henceforth, for any domain that can be written as a union of an FPT number of domains to which we provide FPT algorithms using our framework, our framework automatically yields the same results.

\subsection{$k+\ell$ as Parameters}\label{sec:paramkl}

Here we apply Theorem~\ref{thm:framework_small} to specific domains $\mathcal{D}$.
We design an exact empty extension oracle on $\mathcal{D}$ that runs in FPT time.
As mentioned in Section~\ref{sec:intro_app_small}, in most cases, the exact empty extension oracle on $\mathcal{D}$ can be designed almost directly from an algorithm that finds a set of size $r$ in $\mathcal{D}$. Therefore, most designs for the exact empty extension oracle in this section are straightforward.

\subsubsection{Vertex Cover}
Let $G=(V,E)$ be an undirected graph and $\ell \in \mathbb{Z}_{\geq 0}$.
A vertex subset $Z \subseteq V$ is a \emph{vertex cover} of $G$ if for each $e \in E$, $e \cap Z \neq \emptyset$.
Let $\mathcal{D} \subseteq 2^V$ be the set of all vertex covers of $G$ of size at most $\ell$.
Baste et al.~\cite{baste2019fpt} provided an FPT algorithm parameterized by $k+\ell$ for the max-min/max-sum diversification problems on $\mathcal{D}$.

We design an exact empty extension oracle on $\mathcal{D}$.
Let $r \in \mathbb{Z}_{\geq 0}$ and $Y \subseteq V$.
Define $\mathcal{D}_{Y,r}$ as the family of all vertex covers of size $r$ of $G$ disjoint from $Y$. We construct an algorithm to find a set in $\mathcal{D}_{Y,r}$.
If $Y$ includes two adjacent vertices, $\mathcal{D}_{Y,r}$ is obviously empty.
Otherwise, let $Z$ be the set of vertices adjacent to at least one vertex in $Y$, and let $\mathcal{D}'$ be the family of all vertex covers of size $r-|Z|$ of the subgraph $G'$ induced by $V \setminus (Y \cup Z)$. Then $\mathcal{D}_{Y,r} = \{Z \cup D \colon D \in \mathcal{D}'\}$.
Therefore, it suffices to find a vertex cover of $G'$ of size $r - |Z|$, which can be implemented in FPT time parameterized by $r - |Z|\leq \ell$.

\subsubsection{$t$-Hitting Set}
Let $U$ be a finite set, $\ell, t \in \mathbb{Z}_{\geq 0}$, and $\mathcal{S} \subseteq 2^U$ be a set family such that each $S \in \mathcal{S}$ satisfies $|S| \leq t$.
A subset $Z$ of $U$ is a \emph{hitting set} of $\mathcal{S}$ if for each $S \in \mathcal{S}$, $S \cap Z \neq \emptyset$.
Let $\mathcal{D} \subseteq 2^U$ be the family of all hitting sets of $\mathcal{S}$ of size at most $\ell $.
Baste et al.~\cite{baste2019fpt} provided an FPT algorithm parameterized by $k+\ell+t$ for the max-min/max-sum diversification problems on $\mathcal{D}$.
We design an exact empty extension oracle on $\mathcal{D}$.
Let $Y \subseteq 2^U$.
The problem of finding a hitting set disjoint from $Y$ of size $r \in \mathbb{Z}_{\geq 0}$ is equivalent to finding a hitting set of size $r$ for the set family $\{S \setminus Y \colon S \in \mathcal{S}\}$. Therefore, the oracle can be designed using any FPT algorithm for the hitting set problem.

\subsubsection{Feedback Vertex Set}
Let $G=(V,E)$ be an undirected graph.
A vertex subset $Z \subseteq V$ is a \emph{feedback vertex set} of $G$ if the graph induced by $V \setminus Z$ is acyclic.
Let $\mathcal{D} \subseteq 2^V$ be the set of all feedback vertex sets of $G$.
Baste et al.~\cite{baste2019fpt} provided an FPT algorithm parameterized by $k+\ell $ for the max-min/max-sum diversification problems on $\mathcal{D}$.

We design an exact empty extension oracle on $\mathcal{D}$.
Let $r \in \mathbb{Z}_{\geq 0}$ and $Y \subseteq V$.
Since adding vertices to a feedback vertex set keeps it to be a feedback vertex set, it suffices to find a feedback vertex set of size at most $r$ disjoint from $Y$.
We use the \emph{compact representation} of feedback vertex sets~\cite{guo2006compression} to achieve this.
A family $\mathcal{C}$ of pairwise disjoint subsets of $V$ is said to \emph{represent} a set $Z \subseteq V$ if for each $C \in \mathcal{C}$, $|Z \cap C| = 1$.
When all sets represented by $\mathcal{C}$ are feedback vertex sets of $G$, we call $\mathcal{C}$ a \emph{compact representation of feedback vertex sets}.
A list $\mathfrak{L}$ consisting of compact representations of feedback vertex sets, each of size at most $r$, is called a \emph{compact list of $r$-minimal feedback vertex sets} if, for every inclusion-wise minimal feedback vertex set $Z$ of $G$, there exists $\mathcal{C} \in \mathfrak{L}$ that represents $Z$.
The following is known.
\begin{lemma}[\rm{\cite{guo2006compression}}]
There exists an algorithm that computes a compact list of $r$-minimal feedback vertex sets of size $2^{O(r)}poly(|V|)$ in $2^{O(r)}poly(|V|)$ time.
\end{lemma}

Our algorithm first computes a compact list $\mathfrak{L}$ of $r$-minimal feedback vertex sets.
Then, for each $\mathcal{C} \in \mathfrak{L}$, the algorithm determines whether $\mathcal{C}$ represents a feedback vertex set disjoint from $Y$. This can be acheived by checking the non-emptiness of $C \setminus Y$ for each $C \in \mathcal{C}$. Since every minimal feedback vertex set is represented by some $\mathcal{C} \in \mathfrak{L}$, this procedure gives the desired oracle.

\subsubsection{Matroid Intersection}\label{sec:app_matroid_intersection}
Let $U$ be a finite set, $\ell \in \mathbb{Z}_{\geq 0}$, and $\mathcal{M}_1 := (U, \mathcal{I}_1)$, $\mathcal{M}_2 := (U, \mathcal{I}_2)$ be two matroids.
A subset $Z$ of $U$ is a \emph{common independent set} of $\mathcal{M}_1$ and $\mathcal{M}_2$ if $Z \in \mathcal{I}_1 \cap \mathcal{I}_2$.
Let $\mathcal{D} \subseteq 2^U$ be the set of all common independent sets of $\mathcal{M}_1$ and $\mathcal{M}_2$. 
Fomin et al.~\cite{fomin2024diversecollection} provided an FPT algorithm parameterized by $k+\ell $ for the max-min diversification problems on $\mathcal{D}$.

We design an exact empty extension oracle on $\mathcal{D}$.
Let $Y \subseteq U$.
The problem of finding a common independent set of size $r \in \mathbb{Z}_{\geq 0}$ disjoint from $Y$ is equivalent to finding a common independent set in matroids $\mathcal{M}_1' := (U \setminus Y, \{I \colon I \cap Y = \emptyset \land I \in \mathcal{I}_1\})$ and $\mathcal{M}_2' := (U \setminus Y, \{I \colon I \cap Y = \emptyset \land I \in \mathcal{I}_2\})$. The oracle can thus be designed using any polynomial-time algorithm for the matroid intersection problem.
Furthermore, the following result~\cite{marx2009parameterized} allows extending this result to when $\mathcal{D}$ is a family of common independent sets of $t \in \mathbb{Z}_{\geq 1}$ represented linear matroids.
\begin{lemma}[\rm{\cite{marx2009parameterized}}]
There is an FPT algorithm parameterized by $t$ and $\ell $ for finding a common independent set of size $\ell $ of $t$ given represented linear matroids.
\end{lemma}

\subsubsection{Almost $2$-SAT}\label{sec:app_almostsat}

Let $\varphi$ be a $2$-CNF formula with the set of clauses denoted by $C$ and $\ell\in \mathbb{Z}_{\geq 0}$. 
Let $\mathcal{D} \subseteq 2^C$ be the set of clauses with size at most $\ell $ such that removing those clauses makes $\varphi$ satisfiable. \cite{razgon2009almost} provided an FPT algorithm parameterized by $\ell $ for the \emph{almost $2$-SAT problem}, which is a problem of determining non-emptiness of $\mathcal{D}$. 
We design an exact empty extension oracle on $\mathcal{D}$. Let $Y \subseteq C$. Let $Z$ be the set of variables contained in the clauses of $Y$. We guess the assignments for the variables in $Z$. The cost of this guess is $2^{|Z|} \leq 2^{2|Y|}$. After guessing, the problem is reduced to the original almost $2$-SAT problem.

\subsubsection{Independent Set on Certain Graph Classes}\label{sec:app_mis}

Let $G=(V,E)$ be a graph. A vertex subset $Z \subseteq V$ is an \emph{independent set} of $G$ if $Z$ contains no pair of adjacent vertices. Let $\mathcal{D} \subseteq 2^V$ be the set of all independent sets of $G$ with size at least $\ell \in \mathbb{Z}_{\geq 0}$. The problem of determining the non-emptiness of $\mathcal{D}$ is generally W[1]-hard, but polynomial-time or FPT algorithms are known for several graph classes such as chordal graphs~\cite{gavril1972algorithms}, claw-free graphs~\cite{minty1980maximal}, even-hole-free graphs~\cite{HusicTT19}, and $H$-free graphs for several specific $H$~\cite{bonnet2020parameterized}. 

Let $\mathcal{G}$ be a graph class closed under vertex deletion, and assume an FPT algorithm is known for determining the non-emptiness of $\mathcal{D}$ on $\mathcal{G}$. Let $G \in \mathcal{G}$ and $Y \subseteq V$. Then, the exact empty extension oracle on $\mathcal{D}$ is equivalent to the problem of finding an independent set of size $r$ in the graph $G'$ obtained by removing $Y$ and all edges incident to $Y$ from $G$. Since $\mathcal{G}$ is closed under vertex deletion, we have $G' \in \mathcal{G}$, and thus we have an FPT algorithm for the oracle.

\subsubsection{Reduction from Loss-Less Kernel-Based Framework}

Baste et al.~\cite{baste2022diversity} constructed a framework using loss-less kernels~\cite{carbonnel2016propagation,carbonnel2017kernelization} for the max-sum diversification problem parameterized by $k+\ell $, which also straightforwardly applies to the max-min diversification problem. Here, we introduce a rephrased version of the loss-less kernel using our terminology.
Let $U$ be a finite set and $\mathcal{D} \subseteq 2^U$ be a family such that $|D| \leq \ell $ holds for all $D \in \mathcal{D}$. 
A tuple $(\mathcal{D}', P, Q)$ consisting of two disjoint subsets $P, Q \subseteq U$ and $\mathcal{D}' \subseteq 2^{U \setminus (P \cup Q)}$ is a \emph{loss-less kernel} of $\mathcal{D}$ if $\mathcal{D} = \{D' \cup P \cup Q' \colon D' \in \mathcal{D}' \land Q' \subseteq Q \land |D \cup P \cup Q'| \leq \ell\}$ and $|\mathcal{D}'|$ is bounded by a constant that depends only on $\ell $.

Here, we observe that our result generalizes this framework. We design an exact empty extension oracle using the loss-less kernel.
Let $r \in \mathbb{Z}_{\geq 0}$ and $Y \subseteq U$. 
Let $(\mathcal{D}', P, Q)$ be the loss-less kernel of $\mathcal{D}$.
The oracle exhaustively searches for $D' \in \mathcal{D}'$ and determines whether $Y \cap (D' \cup P) = \emptyset$. If so, and if $|D' \cup P| \leq r \leq |D' \cup P \cup (Q \setminus Y)|$, a desired set is found. If no desired set is found for any $D' \in \mathcal{D}'$, the oracle returns $\bot$.

\subsubsection{Reduction from Color-Coding-Based Framework}

Hanaka et al.~\cite{hanaka2021finding} constructed a color-coding~\cite{alon1995color} based framework to provide FPT algorithms for the max-min diversification problems parameterized by $k+\ell $. Specifically, they used the following oracle to solve the max-min diversification problem.
\begin{quote}
\textbf{Colorful Set Oracle on $\mathcal{D}$:}
Let $U$ be a finite set, $s, t \in \mathbb{Z}_{\geq 1}$ with $s \leq t$, $c \in \{1, \dots, t\}^U$, and $\mathcal{D} \subseteq 2^U$.
It determines whether there exists a set $D \in \mathcal{D}$ of size $s$ such that $\{c_e \colon e \in D\} = \{1, \dots, s\}$, and if so, finds one.
\end{quote}
They proved that there is an oracle algorithm for the max-min diversification problem on $\mathcal{D}$ using the colorful set oracle on $\mathcal{D}$, where the number of oracle calls and the time complexity are both FPT parameterized by $k+\ell $. 
They applied this framework on several domains $\mathcal{D}$, including the family of paths, interval schedulings, and matchings of size $\ell $.

Here, we observe that our framework generalizes theirs. It suffices to design an exact empty extension oracle on $\mathcal{D}$ using the colorful set oracle on $\mathcal{D}$.
Let $U$ be a finite set, $r \in \mathbb{Z}_{\geq 0}$, and $Y \subseteq U$.
Let $s = r, t = r+1$, and for each $e \not\in Y$, take $c_e \in \{1, \dots, r\}$ independently and uniformly at random. For $e \in Y$, let $c_e = r+1$.
If $D\in \mathcal{D}$ satisfies $|D|=r$ and $D\cap Y=\emptyset$, $\{c_e \colon e \in D\} = \{1, \dots, s\}$ holds with probability $\frac{1}{r!}$. Thus, by calling the colorful set oracle on $\mathcal{D}$ $O(r!)$ times, we can implement the exact empty extension oracle that works correctly with sufficiently high probability.
Furthermore, we can derandomize this algorithm using a perfect hash family~\cite{alon1995color}.

\subsection{$k+d$ as Parameters}\label{sec:paramkd}
Here we apply Theorem~\ref{thm:framework_large} to specific domains $\mathcal{D}$.
For each domain $\mathcal{D}$, we design the $(-1,1)$-optimization oracle on $\mathcal{D}$ and the exact extension oracle on $\mathcal{D}$.

\subsubsection{Matroid Base and Branching}\label{sec:app_branch}

Let $\mathcal{M} = (U, \mathcal{I})$ be a matroid and $\mathcal{D}$ be the family of bases of $\mathcal{M}$.
Fomin et al.~\cite{fomin2024diversecollection} provided an FPT algorithm parameterized by $k+d$ for the max-min diversification problem on $\mathcal{D}$.
We design the $(-1, 1)$-optimization oracle and the exact extension oracle on $\mathcal{D}$ to obtain an FPT algorithm for diversification and clustering problems on $\mathcal{D}$.

The $(-1,1)$-optimization oracle on $\mathcal{D}$ is just a weighted maximization problem on matroid bases and can be implemented in polynomial time using a greedy algorithm.
We provide a polynomial-time algorithm for the exact extension oracle on $\mathcal{D}$.
Let $r \in \mathbb{Z}_{\geq 0}$, $X, Y \subseteq U$, and $C \in \mathcal{D}$.
If $r$ is odd, there is clearly no solution, so we assume $r$ is even.
The algorithm first computes a set $D_{\min} \in \mathcal{D}$ (resp. $D_{\max}\in \mathcal{D}$) that minimizes (resp. maximizes) $|D \triangle C|$ among $D\in \mathcal{D}$ that contain $X$ and disjoint from $Y$. 
This can be formulated as a weighted maximization problem on a matroid and implemented using a greedy algorithm. If $r < |D_{\min} \triangle C|$ or $|D_{\max} \triangle C| < r$, there is no solution.
Otherwise, we can find a solution using \emph{strong exchange property} of $\mathcal{D}$.
Here, a set family $\mathcal{Z}$ satisfies the \emph{strong exchange property} if for any distinct $Z_1, Z_2 \in \mathcal{Z}$ and any $e_1 \in Z_1 \setminus Z_2$, there exists $e_2 \in Z_2 \setminus Z_1$ such that $Z_1 \setminus \{e_1\} \cup \{e_2\} \in \mathcal{Z}$.
It is well-known that the base family of a matroid satisfies the strong exchange property.
Starting from $D_1 = D_{\min}$ and $D_2 = D_{\max}$, the algorithm repeatedly applies the strong exchange property on $D_1$ and $D_2$ and replaces $D_1$ by the obtained solution.
Then, $|D_1\triangle C|=r$ eventually holds after at most $\frac{|D_{\min}\triangle D_{\max}|}{2}$ iterations.

We can easily see that this algorithm actually works under a weaker assumption, \emph{weak exchange property} of $\mathcal{D}$.
Here, $\mathcal{Z}$ satisfies the \emph{weak exchange property} if for any distinct $Z_1, Z_2 \in \mathcal{Z}$, there exist elements $e_1 \in Z_1 \setminus Z_2$ and $e_2 \in Z_2 \setminus Z_1$ such that $Z_1 \setminus \{e_1\} \cup \{e_2\} \in \mathcal{Z}$.
Particularly, consider the case that $\mathcal{D}$ is the family of \emph{$r$-branching} in a graph. 
Here, for a directed graph $G = (V, E)$ and $r \in V$, an edge subset $T \subseteq E$ is an \emph{$r$-branching} if $|T| = |V| - 1$ and for all vertices $v \in V$, there is an $r-v$ path consisting only of edges in $T$.
The $(-1,1)$-optimization oracle on $\mathcal{D}$ can be implemented using any polynomial-time algorithm for finding a maximum weight $r$-branching.
Since it is well-known that $\mathcal{D}$ satisfies weak exchange property~\cite{ito2023reconfiguring}, we can design an exact extension oracle that works in polynomial time.
Furthermore, when $\mathcal{D}$ is the family of \emph{branchings} in a graph $G$, that is, the family of edge subsets forming an $r$-branching for some $r \in V$, $\mathcal{D}$ can be written as the union of the families of $r$-branchings for each $r \in V$. Hence, by the observation at the beginning of this section, the same result holds for $\mathcal{D}$ as well.

Given a directed graph $G=(V,E)$ and $s,t \in V$, the \emph{$d$-distinct branchings} problem is to determine whether there exist an $s$-branching $B_1$ and a $t$-in-branching $B_2$ of $G$ such that $|B_1 \triangle B_2| \geq d$, where an edge subset is an \emph{$r$-in-branching} if it is an $r$-branching in the graph obtained by reversing the direction of every edge in $G$.
Bang-Jensen et al.~\cite{bang2016parameterized} provided an FPT algorithm for this problem when $G$ is strongly connected. Gutin et al.~\cite{gutin2018k} extended their results to general graphs, Bang-Jensen et al.~\cite{bang2021k} obtained a polynomial kernel and improved the time complexity, and Eiben et al.~\cite{eiben2024determinantal} further improved the time complexity.
By the discussion in Section~\ref{sec:usage_diverse}, we can obtain another FPT algorithm for this problem using $d$-limited $k$-max-distance sparsifiers. 
Specifically, the algorithm computes the $d$-limited $1$-max-distance sparsifiers $\mathcal{K}_1$ and $\mathcal{K}_2$ of a family of $s$-branchings and $t$-in-branchings, respectively, with respect to $2^U$, and then performs an exhaustive search over pairs $(B_1, B_2) \in \mathcal{K}_1 \times \mathcal{K}_2$.
We can naturally extend this algorithm to the case where we choose multiple branchings and in-branchings. 
Specifically, we obtain an FPT algorithm parameterized by $k+d$ for the following problem. 
Let $k \in \mathbb{Z}_{\geq 1}$, $d \in \mathbb{Z}_{\geq 0}$, $k'$ be an integer with $0 \leq k' \leq k$, and $s_1, \dots, s_k \in V$. 
Let $\mathcal{D}_1, \dots, \mathcal{D}_{k'}$ be the family of $s_1, \dots, s_{k'}$-branchings, respectively, and $\mathcal{D}_{k'+1}, \dots, \mathcal{D}_k$ be the family of $s_{k'+1}, \dots, s_k$-in-branchings, respectively. 
The problem asks whether there is a $k$-tuple $(B_1, \dots, B_k) \in \mathcal{D}_1 \times \dots \times \mathcal{D}_k$ such that $\min_{1 \leq i < j \leq k} |D_i \triangle D_j| \geq d$ (or $\sum_{1 \leq i < j \leq k} |D_i \triangle D_j| \geq d$).


\subsubsection{Matching}\label{sec:app_matching}

Let $G = (V, E)$ be an undirected graph. An edge subset $M \subseteq E$ is a \emph{matching} of $G$ if no two edges in $M$ share a vertex.
A matching $M$ of $G$ is \emph{perfect} if $2|M| = |V|$.
Let $\ell\in  \left\{0, \dots, \floor{\frac{|V|}{2}}\right\}$ and $\mathcal{D}$ be the family of matchings in $G$ consists of $\ell $ edges.
Fomin et al.~\cite{fomin2024diversepair} provided FPT algorithms parameterized by $d$ for the max-min diversification problem on $\mathcal{D}$ for the case $k=2$ and $\ell $ is the size of the maximum matching. 
Fomin et al.~\cite{fomin2024diversecollection} provided the same result parameterized by $k+d$ for general $k$ for the case $2\ell=|V|$.
Eiben et al.~\cite{eiben2024determinantal} generalize the result for the perfect matching to general $\ell $, as well as improving the time complexity.
Let $d \in \mathbb{Z}_{\geq 0}$ and $k \in \mathbb{Z}_{\geq 1}$.
In this section, we design the $(-1,1)$-optimization oracle and the exact extension oracle on $\mathcal{D}$.
The following lemma is useful.
\begin{lemma}\label{lem:match_perfectize}
Let $W$ a vertex set of size $|V| - 2\ell $ disjoint from $V$.
Let $F := \{(v, w) \colon v \in V, w \in W\}$ and $G^l := (V \dot{\cup} W, E \dot{\cup} F)$. Then, if $M^l\subseteq E\dot{\cup} F$ is a perfect matching in $G^\ell $, $M^l \cap E$ is a matching of size $\ell $ in $G$. Conversely, if there exists a matching of size $\ell $ in $G$, there exists a perfect matching in $G^\ell $.
\end{lemma}
\begin{proof}
Let $M^\ell $ be a perfect matching in $G^\ell $. Since $W$ is an independent set, $M^\ell $ consists of $|V| - 2\ell $ edges between $V$ and $W$ and $\ell $ edges within $V$, which means $M^l \cap E$ is a matching of size $\ell $ in $G$. Conversely, a matching of size $\ell $ in $G$ can be naturally extended to a perfect matching in $G^\ell $.
\end{proof}
We call $G^\ell $ the \emph{$\ell $-expanded graph} of $G$.
We now design a $(-1,1)$-optimization oracle on $\mathcal{D}$ that runs in polynomial time.
Let $w\in \{-1,1\}^{E}$.
The algorithm first constructs the $\ell $-expanded graph $G^l := (V \dot{\cup} W, E \dot{\cup} F)$ of $G$. Define the edge weights $\widehat{w}\in \{-1,0,1\}^{E\dot{\cup}F}$ by $\widehat{w}_e = w_e$ for $e \in E$ and $\widehat{w}_e = 0$ for $e \in F$. Then, the algorithm finds the maximum weight perfect matching $M^\ell $ with respect to $\widehat{w}$.
From Lemma~\ref{lem:match_perfectize}, $M^l \cap E$ is the desired output for the $(-1,1)$-optimization oracle on $\mathcal{D}$.

We now design an exact extension oracle on $\mathcal{D}$ that runs in polynomial time.
Let $r \in \mathbb{Z}_{\geq 0}$, $X, Y \subseteq E$, and $C \in \mathcal{D}$.
Let $G' := (V' := V \setminus V[X], E' := E \setminus (X \cup Y))$, where $V[X]$ denotes the set of endpoints of edges in $X$.
Color each edge $e \in E'$ red if $e \in C$ and blue if $e \notin C$.
The problem reduces to determining whether there exists a matching of size $\ell - |X|$ in $G'$ that contains exactly $r - |X \setminus C|$ blue edges.
Let $G'' = (V' \cup W', E' \cup F')$ be the $(l - |X|)$-expanded graph of $G'$, and color the edges in $F'$ red. The problem reduces to determining whether there exists a perfect matching in $G''$ containing exactly $r - |X \setminus C|$ blue edges.
This problem is known as the \emph{exact matching problem}~\cite{mulmuley1987matching} and admits a randomized polynomial-time algorithm.

\subsubsection{Minimum Edge Flow}\label{sec:app_flow}

Let $G = (V, E)$ be a graph.
For simplicity, we consider only the case when $G$ is directed; the results in this part naturally extend to undirected graphs.
Let $s, t \in V$ and $b \in \mathbb{Z}_{\geq 0}$. An edge subset $F \subseteq E$ is a \emph{$s,t$-flow of amount $b$} if $F$ can be written as the union of the edge sets of $b$ edge-disjoint $s,t$-paths.
Let $\mathcal{D}$ be the family of $s,t$-flows of amount $b$ with the minimum number of edges.
In this section, we design the $(-1,1)$-optimization oracle on $\mathcal{D}$ and the exact extension oracle on $\mathcal{D}$.

We provide a $(-1,1)$-optimization oracle on $\mathcal{D}$ that runs in polynomial time.
Let $w \in \{-1,1\}^E$.
Define edge weights $\widehat{w} \in \mathbb{Z}^E$ by $\widehat{w}_e := 2|E| - w_e$.
The algorithm finds the minimum weight $s,t$-flow $F$ of amount $b$ with respect to $\widehat{w}$. The correctness of the algorithm is shown as follows.
\begin{lemma}\label{lem:flow_opt}
$F$ is the desired output for the $(-1,1)$-optimization oracle on $\mathcal{D}$.
\end{lemma}
\begin{proof}
It is sufficient to state that $F$ is an $s,t$-flow of amount $b$ with the minimum number of edges. Assume the contrary and let $F'$ be an $s,t$-flow of amount $b$ with the minimum number of edges. Then, we have $\widehat{w}(F') \leq 2|E||F'| + |F'| \leq 2|E||F| - (2|E|-|F'|) < 2|E||F| - |F| \leq \widehat{w}(F)$, contradicting the minimality of $\widehat{w}(F)$. 
\end{proof}

We provide an exact extension oracle on $\mathcal{D}$ that runs in FPT time. 
Let $r \in \mathbb{Z}_{\geq 0}$, $X, Y \subseteq E$, and $F \in \mathcal{D}$.
Let $G' = (V, E')$ be the graph obtained from $G$ by reversing the direction of edges in $F$.
For an edge set $Z \subseteq E'$ and a vertex $v \in V$, let $\delta^{+}_Z(v)$ denote the set of edges in $Z$ leaving $v$ and $\delta^{-}_Z(v)$ denote the set of edges in $Z$ entering $v$.
An edge set $Z$ of a directed graph is \emph{Eulerian} if for all $v \in V$, $|\delta^{+}_Z(v)| = |\delta^{-}_Z(v)|$.
Let $\mathcal{D}^*$ be the family of Eulerian edge sets $D^*$ such that $|D^* \cap F| = |D^* \setminus F|$.
Then $\mathcal{D} = \{F \triangle D^* \colon D^* \in \mathcal{D}^*\}$.
Thus, to design the exact extension oracle on $\mathcal{D}$, it suffices to determine whether there exists an Eulerian edge subset $D^* \subseteq E'$ such that $|D^* \cap F| = |D^* \setminus F| = \frac{r}{2}$, $X' \subseteq Z$, and $Y' \cap Z = \emptyset$, where $X':=(X \setminus F) \cup (Y \cap F)$ and $Y':=(X \cap F) \cup (Y \setminus F)$.
Using the fact that Eulerian edge sets can be decomposed into edge-disjoint cycles, we solve this problem via color-coding.
For each $e \in E'$, sample $c_e \in \{1, \dots, r\}$ independently and uniformly at random. If there exists a $D^* \in \mathcal{D}^*$ meeting the conditions, there is a $\frac{1}{r!}$ probability that $\{c_e \colon e \in D^*\} = \{1, \dots, r\}$.
Define $\DP[u][v][C][p][q]$ by $Z$ if there exists a $Z \subseteq E' \setminus Y'$ that satisfies the following conditions, and $\bot$ otherwise.
\begin{itemize}
    \item $|\delta^{+}_Z(w)| - |\delta^{-}_Z(w)|$ is equal to $1$ if $w = u \neq v$, $-1$ if $w = v \neq u$, and $0$ otherwise,
    \item $|Z|=|C|$ and $\{c_e \colon e \in Z\} = C$,
    \item $|Z \cap X'| = p$, and
    \item $|Z \cap F| = q$.    
\end{itemize}
If there exists $v\in V$ such that $\DP[v][v][\{1,\dots, r\}][|X'|][r/2] \neq \bot$, then this is the desired output for the oracle.
The algorithm initializes the table with
\begin{align*}
    \DP[u][v][C][p][q] := \begin{cases}
        \emptyset & (u=v, C=\emptyset, p=q=0)\\
        \bot & (\text{otherwise})
    \end{cases}
\end{align*}
and updates the table by the following rules.
\begin{itemize}
    \item If $\DP[u][v][C][p][q]\neq \bot$, set $\DP[u][w][C \cup \{c_e\}][p'][q']:=\DP[u][v][C][p][q]\cup \{e\}$ for each $w\in V$ such that there exists an edge $e:=(v, w) \in E'$ with $c_e \not \in C$, where $p'$ equals $p + 1$ if $e \in X'$ and $p$ otherwise, and $q'$ equals $q + 1$ if $e \in F$ and $q$ otherwise.
    \item If $\DP[v][v][C][p][q]\neq \bot$, set $\DP[w][w][C][p][q]:=\DP[v][v][C][p][q]$ for each $w \in V$.
\end{itemize}
This algorithm runs in FPT time parameterized by $r$.

\subsubsection{Minimum Edge Steiner Tree}\label{sec:app_steiner}

Let $G=(V,E)$ be an undirected graph and $T\subseteq V$.
An edge subset $F\subseteq E$ is called a \emph{Steiner tree} with a terminal set $T$ if the graph $\left(V_F:=\bigcup_{e\in F}e, F\right)$ is connected and $T\subseteq V_F$.
Let $\mathcal{D}$ be the family of all Steiner trees with the minimum number of edges.
In this section, we design the $(-1,1)$-optimization oracle on $\mathcal{D}$ and the exact extension oracle on $\mathcal{D}$.
Funayama et al.~\cite{funayama2024parameterized} provided an FPT algorithm for the max-min diversification problem on the domain of shortest paths parameterized by $k+d$.
Their problem for undirected graphs corresponds to the case of $|T|=2$.

We provide a $(-1,1)$-optimization oracle on $\mathcal{D}$ that runs in FPT time parameterized by $|T|$. This is almost identical to the textbook algorithm~\cite{cygan2015parameterized}.
Let $w\in \{-1,1\}^E$.
For each $v\in V$ and $S\subseteq T$, let $\OPT[v][S]$ be the minimum number of edges in a Steiner tree with terminal set $S\cup \{v\}$, and $\DP[v][S]$ be the one with the maximum weight among such Steiner trees. What we seek is $\DP[v][T]$ for any $v\in T$.
This dynamic programming algorithm initializes
\begin{align*}
    \OPT[v][S]:=\begin{cases}
        0 & (v\in T, S=\{v\})\\
        \infty & (\text{otherwise})
    \end{cases},\quad \quad 
    \DP[v][S]:=\begin{cases}
        \emptyset & (v\in T, S=\{v\})\\
        \bot & (\text{otherwise})
    \end{cases}
\end{align*}
and updates using
\begin{align*}
    \OPT[v][S]&:=\min\left(\min_{u\in V\colon (u,v)\in E}\OPT[u][S]+1,
    \min_{\emptyset \neq S'\subsetneq S}\left(\OPT[v][S']+\OPT[v][S\setminus S']\right)\right),
\end{align*}
to compute $\OPT[\cdot][\cdot]$ and sets $\DP[v][S]$ to be the non-$\bot$ set with the maximum weight among the following sets, if any; otherwise, $\bot$. 
\begin{itemize}
    \item $\DP[u][S]\cup \{(u,v)\}$ for each $u\in V$ such that $(u,v)\in E$ and $\OPT[u][S]+1=\OPT[v][S]$.
    \item $\DP[v][S']\cup \DP[v][S\setminus S']$ for each $\emptyset\neq S'\subsetneq S$ such that $\OPT[v][S']+\OPT[v][S\setminus S']=\OPT[v][S]$.
\end{itemize}
Here, for simplicity, we set $\bot\cup Z = \bot$ for any set $Z$.
This runs in FPT time parameterized by $|T|$.

We provide an exact extension oracle on $\mathcal{D}$ that runs in FPT time parameterized by $|T|$. This is also almost identical to the same textbook algorithm.
Let $r\in \mathbb{Z}_{\geq 0}$, $X,Y\subseteq E$, and $F\in \mathcal{D}$.
For each $v\in V$, $S\subseteq T$, $p\in \{0,\dots, |X|\}$, and $q\in \{0,\dots, r/2\}$, define $\EX[v][S][p][q]$ as the Steiner tree with $\OPT[v][S]$ edges with terminal set $S\cup \{v\}$, disjoint from $Y$, including $p$ edges in $X$, and including $q$ edges not in $F$, if such Steiner tree exists; otherwise, it is $\bot$. The solution required for the oracle is $\EX[v][T][|X|][r/2]$ for any $v\in T$.
This dynamic programming algorithm initializes
\begin{align*}
    \EX[v][S][p][q]:=\begin{cases}
        \emptyset & (v\in T, S=\{v\}, p=q=0)\\
        \bot & (\text{otherwise})
    \end{cases}
\end{align*}
and sets $\EX[v][S][p][q]$ to be the non-$\bot$ sets among the following if there exists any; otherwise, $\bot$.
\begin{itemize}
    \item $\EX[u][S][p'][q']\cup \{(u,v)\}$ for each $u\in V$ such that $(u,v)\in E$ and $\OPT[u][S]+1=\OPT[v][S]$. Here, $p'$ equals $p-1$ if $(u,v)\in X$, and $p$ otherwise; $q'$ equals $q-1$ if $(u,v)\not \in F$, and $q$ otherwise.
    \item $\EX[v][S'][p'][q']\cup \EX[v][S\setminus S'][p-p'][q-q']$ for each $\emptyset\neq S'\subsetneq S$ with $\OPT[v][S']+\OPT[v][S\setminus S']=\OPT[v][S]$, $p'\in \{0,\dots, p\}$, and $q'\in \{0,\dots, q\}$.
\end{itemize}
This dynamic programming algorithm runs in FPT time parameterized by $|T|$.

\subsubsection{Vertex Set of Minimum $s,t$-Cut}\label{sec:app_cut}

Let $G=(V,E)$ be a directed or undirected graph and $s,t\in V$.
A vertex subset $C\subseteq V$ is an \emph{$s,t$-cut} if $s\in C$ and $t\not \in C$.
The \emph{cost} of the $s,t$-cut $C$ is $|\delta(C)|:=|\{e\in E\colon |e\cap C|=1\}|$.
Let $\mathcal{D}\subseteq 2^V$ be the family of minimum cost $s,t$-cuts.
Let $k\in \mathbb{Z}_{\geq 1}$ and $d\in \mathbb{Z}_{\geq 0}$.
In this section, we construct a $d$-limited $k$-max-distance sparsifier of $\mathcal{D}$.

We provide an $(-1,1)$-optimization oracle that runs in polynomial time. Let $w\in \{-1,1\}^V$.
We construct the edge set $\widehat{E}$ consisting of 
\begin{itemize}
    \item for each $e\in E$, an edge $e$ with weight $2|V|+1$,
    \item for each $v\in V$ with $w_v=1$, an edge $(s,v)$ with weight $1$, and
    \item for each $v\in V$ with $w_v=-1$, an edge $(v,t)$ with weight $1$.
\end{itemize}
Then, solve the minimum weight $s,t$-cut problem on the graph $(V,\widehat{E})$.
We can easily see that the $s,t$-cut returned by this procedure is indeed the solution required by the oracle by the similar argument as the proof of Lemma~\ref{lem:flow_opt}.

Let $C\in \mathcal{D}$, $r\in \mathbb{Z}_{\geq 0}$ and $X,Y\subseteq U$. 
Unfortunately, we do not have an exact extension oracle itself; however, we can construct an FPT algorithm that either outputs a trivial $d$-limited $k$-max-distance sparsifier or behaves as an exact extension oracle.
Given that our goal in this section is to compute the $d$-limited $k$-max-distance sparsifier, this is sufficient.
Furthermore, if this procedure returns a trivial $d$-limited $k$-max-distance sparsifier, it implies there is no solution to the $k$-center or $k$-sum-of-radii clustering problems. 
Therefore, algorithms for clustering problems can also be obtained without any issues.

A finite set equipped with a binary relation $(W,\preceq)$ is a \emph{poset} if $\preceq$ is reflexive, antisymmetric, and transitive.
For $v\in W$, define $I^{+}(v)=\{u\in W\colon u\succeq v\}$ and $I^{-}(v)=\{u\in W\colon u\preceq v\}$.
A subset $I\subseteq W$ is called an \emph{ideal} if for any $w\in I$, $I^{-}(w)\subseteq I$.
An ideal $I$ is \emph{proper} if $I$ is neither $\emptyset$ nor $W$.
The following is well-known.
\begin{lemma}[\rm{\cite{birkhoff1937rings,BergMS23}}]\label{lem:birkhoff}
There is a poset $\mathcal{P}=(W,\preceq)$ and a list of nonempty disjoint subsets $\{A_w\}_{w\in W}$ of $U$ such that $\mathcal{D}=\bigcup_{\text{$I$ is proper ideal of $\mathcal{P}$}}\left\{\bigcup_{w\in I}A_w\right\}$. Such a poset $\mathcal{P}$ can be computed in polynomial time.
\end{lemma}

We take the poset $\mathcal{P}=(W,\preceq)$ and $\{A_w\}_{w\in W}$ in Lemma~\ref{lem:birkhoff} and the proper ideal $I_C$ corresponding to $C$.
Let $p:=(4d+2)^2 2^{k-1}$ be the upper bound of $r$ among all calls of the exact extension oracle in our framework.
Let $W^{+}_{C}:=\left\{w\in W\setminus I_C\colon |I^{-}(w)\setminus I_C|\leq p\right\}$ and $W^{-}_{C}:=\left\{w\in I_C\colon |I^{+}(w)\cap I_C|\leq p\right\}$.
We have the following.
\begin{lemma}\label{lem:cut_minusplus}
Let $D\in \mathcal{D}$ and assume $|D\triangle C|\leq p$. Let $I_{D}$ be the ideal of $\mathcal{P}$ corresponding to $D$. Then, $I_C\setminus W^{-}_{C}\subseteq I_{D}\subseteq I_C\cup W^{+}_C$.
\end{lemma}
\begin{proof}
Assume there is an element $w$ in $I_{D}\setminus (I_C\cup W^{+}_C)$. From the definition of $W^{+}_{C}$, we have $|I^{-}(w)\setminus I_C| > p$. Therefore, $|D\triangle C|\geq \sum_{w\in I_{D}\setminus I_{C}}|A_w|\geq |I_{D}\setminus I_C|\geq |I^{-}(w)\setminus I_C| > p$.
Similarly, assume there is an element $w$ in $(I_C\setminus W^{-}_C)\setminus I_{D}$. From the definition of $W^{-}_{C}$, we have $|I^{+}(w)\cap I_C| > p$. Therefore, $|D\triangle C|\geq \sum_{w\in I_{C}\setminus I_{D}}|A_w|\geq |I_{C}\setminus I_{D}|\geq |I_C \cap I^{+}(w)| > p$.
\end{proof}

Since the size of $W^{+}_{C}$ or $W^{-}_{C}$ can be large in general, Lemma~\ref{lem:cut_minusplus} does not directly lead to an FPT algorithm for the exact extension oracle. However, we can claim the following.
\begin{lemma}
Assume $|W^{+}_{C}| > 2kd$. Then, a $d$-limited $k$-max-distance sparsifier of $\mathcal{D}$ can be computed in polynomial time. The same holds for $|W^{-}_{C}|$.
\end{lemma}
\begin{proof}
Order elements of $W^{+}_{C}$ as $w_1,\dots, w_{|W^{+}_{C}|}$ so that $w_j\preceq w_{j'}$ implies $j\leq j'$. For each $j\in \{0,\dots, |W^{+}_{C}|-1\}$, $I_C\cup \{w_1,\dots, w_j\}$ is a proper ideal of $\mathcal{P}$. 
For $i\in \{0,\dots, k\}$, let $C_i:=C\cup \bigcup_{j\in \{1,\dots, 2id\}}A_{w_j}$. Then, $\{C_0,\dots, C_k\}$ is a $d$-limited $k$-max-distance sparsifier by Lemma~\ref{lem:tooscattered}. The same argument applies to $W^{-}_{C}$.
\end{proof}
Thus, we can assume $|W^{+}_{C}|, |W^{-}_{C}| \leq 2kd$.
Then, Lemma~\ref{lem:cut_minusplus} implies that there is at most $2^{4kd}$ sets $D\in \mathcal{D}$ with $|D\triangle C|\leq p$, and thus the exact extension oracle can be designed by a brute-force search.

\subsubsection{Vertex Set of Edge Bipartization}\label{sec:app_bipartization}

Let $G=(V,E)$ be an undirected graph.
For $D\subseteq V$, the \emph{edge bipartization set of $D$} is defined by $\beta(D) := E\setminus \delta(D)$.
Let $s:=\min_{D\subseteq U}|\beta(D)|$ and $\mathcal{D}$ be the family of all $D\subseteq V$ with $|\beta(D)|=s$.
In this section, we design the $(-1,1)$-optimization oracle on $\mathcal{D}$ and the exact extension oracle on $\mathcal{D}$.

Take any $D\in \mathcal{D}$.
We use the technique in~\cite{guo2006compression} for designing iterative compression algorithms for the edge bipartition problem to characterize the domain $\mathcal{D}$ using $D$.
For $D_1,D_2\subseteq V$, we denote $E(D_1,D_2):=\{(u,v)\in E\colon u\in D_1\land v\in D_2\}$.
We begin by the following.
\begin{lemma}\label{lem:ftosym}
Let $D_1,D_2\subseteq V$.
Then, $\beta(D_1)\triangle \beta(D_2)=\delta(D_1\triangle D_2)$.
\end{lemma}
\begin{proof}
Then, we have $\beta(D_1)\setminus \beta(D_2)= E(D_1\setminus D_2,D_1\cap D_2) \cup E(D_2\setminus D_1,V\setminus (D_1\cup D_2))$. 
Similarly, $\beta(D_2)\setminus \beta(D_1)= E(D_2\setminus D_1,D_1\cap D_2) \cup E(D_1\setminus D_2,V\setminus (D_1\cup D_2))$. Thus, $\beta(D_1)\triangle \beta(D_2)=E(D_1\triangle D_2,D_1\cap D_2)\cup E(D_1\triangle D_2, V\setminus (D_1\cup D_2))=E(D_1\triangle D_2,V\setminus (D_1\triangle D_2))=\delta(D_1\triangle D_2)$.
\end{proof}

For $A,B\subseteq U$ and $t\in \mathbb{Z}_{\geq 0}$, define $\mathrm{MinCut}(A,B,t)$ as the family of all sets $C\subseteq V$ with $A\subseteq C$, $B\cap C=\emptyset$, and $|\delta(C)|=t$ if $t$ is the minimum possible value of $|\delta(C)|$; otherwise, define it as $\emptyset$.
Furthermore, for $Z\subseteq E$, let $V[Z]:=\bigcup_{e\in Z}e$.
We have the following.
\begin{lemma}\label{lem:edgebipartite_cut}
For any $D\in \mathcal{D}$,
\begin{align*}
    \mathcal{D}=\bigcup_{T\subseteq V[\beta(D)]}\left\{D\triangle C\colon C\in \mathrm{MinCut}(T,V[\beta(D)]\setminus T, 2|\beta(D)\cap \delta(T)|)\right\}.
\end{align*}
\end{lemma}
\begin{proof}
Let $D'\subseteq U$.
Then, we have 
\begin{align*}
    |\beta(D')|=|\beta(D)|-2|\beta(D)\setminus \beta(D')|+|\beta(D)\triangle \beta(D')|=|\beta(D)|-2|\beta(D)\setminus \beta(D')|+|\delta(D\triangle D')|.
\end{align*}
Let $T\subseteq V[\beta(D)]$ and assume $V[\beta(D)]\cap (D\triangle D')=T$.
Then, we have 
\begin{align*}
    \beta(D)\setminus \beta(D')=E(D\setminus D',D\cap D')\cup E(D'\setminus D, V\setminus (D\cup D')) = \beta(D)\cap \delta(D\triangle D') = \beta(D)\cap \delta(T),
\end{align*}
where the last equality is from the fact that 
\begin{align*}
    F\cap \delta(Z)=F\cap ((\delta(V[F]\cap Z)\setminus E(V[F]\cap Z, Z\setminus V[F]))\cup E(Z\setminus V[F],V\setminus Z))=F\cap \delta(V[F]\cap Z)
\end{align*}
holds for all $F\subseteq E$ and $Z\subseteq V$.
By the minimality of $|\beta(D)|$, $|\beta(D')|=|\beta(D)|=s$ holds if and only if $|\delta(D\triangle D')|$ takes the minimum value $2|\beta(D)\cap \delta(T)|$.
\end{proof}
Lemma~\ref{lem:edgebipartite_cut} implies that optimization problems on $\mathcal{D}$ can be reduced to solving the corresponding optimization problems on the domain $\mathrm{MinCut}(T,V[\beta(D)]\setminus T, t)$ for each of $2^{|\beta(D)|}\leq 4^s$ candidates of $T\subseteq \beta(D)$, where $t=2|\beta(D)\cap \delta(T)|\leq 2s$.

We provide the $(-1,1)$-optimization oracle on $\mathcal{D}$ that runs in FPT time parameterized by $s$. Let $w\in \{-1,1\}^V$.
From Lemma~\ref{lem:edgebipartite_cut}, given two disjoint vertex subsets $A,B\subseteq V$ and $t\leq 2s$, it is sufficient to find a set that maximizes $w(C)$ among $C\in \mathrm{MinCut}(A,B,t)$.
The algorithm constructs a graph $(\widehat{V},\widehat{E})$, where $\widehat{V}=V\dot{\cup} \{a,b\}$ and $\widehat{E}$ consists of the following edges.
\begin{itemize}
    \item For each $e\in E$, an edge $e$ with weight $2|V|+1$.
    \item For each $a'\in A$, an edge $(a,a')$ with weight $\infty$.
    \item For each $b'\in B$, an edge $(b',b)$ with weight $\infty$.
    \item For each $v\in V$ with $w_v=1$, an edge $(a,v)$ with weight $1$.
    \item For each $v\in V$ with $w_v=-1$, an edge $(v,b)$ with weight $1$.
\end{itemize}
Then, the algorithm solve the minimum weight $a,b$-cut problem on the graph $(\widehat{V},\widehat{E})$.
If the weight of the returned cut is in the range $[(2|V|+1)t-|V|,(2|V|+1)t+|V|]$, it maximizes $w(C)$ among $C\in \mathrm{MinCut}(A,B,t)$. Otherwise, $\mathrm{MinCut}(A,B,t)=\emptyset$.

Let $r\in \mathbb{Z}_{\geq 0}$, $X,Y\subseteq V$, and $D\in \mathcal{D}$. 
We provide an exact extension oracle on $\mathcal{D}$ that runs in FPT time parameterized by $s$ and $r$.
From Lemma~\ref{lem:edgebipartite_cut}, given two disjoint subsets $A,B\subseteq V$ and $t\leq 2s$, it is sufficient to find a set in $C\in \mathrm{MinCut}(A,B,t)$ such that $|C| = r$, $X\subseteq D\triangle C$, and $Y\cap (D\triangle C)=\emptyset$.
Without loss of generality, we can assume $\mathrm{MinCut}(A,B,t)\neq \emptyset$.
The algorithm first computes $C\in \mathrm{MinCut}(A,B,t)$ that is inclusion-wise minimal. 
It is well-known that such a cut is unique; specifically, it is the cut corresponding to the ideal $\emptyset$ in Lemma~\ref{lem:birkhoff}. 
If $|C|>r$, there is no set $C$ satisfying the condition, and the oracle returns $\bot$. 
If $|C|=r$, the algorithm checks the conditions $X\subseteq D\triangle C$ and $Y\cap (D\triangle C)=\emptyset$, and if satisfied, outputs $C$; otherwise, it returns $\bot$.
Assume $|C| < r$. Since $C$ is unique inclusion-wise minimal set in $\mathrm{MinCut}(A,B,t)$, for any $C'\in \mathrm{MinCut}(A,B,t)\setminus \{C\}$, there exist an edge $e\in \delta(C)$ with $C'\in \mathrm{MinCut}(A\cup e,B,t)$. 
For each edge $e\in \delta(C)$, the algorithm recurses on the instance where $A$ is replaced by $A\cup e$ to compute the desired set $C$.
Since $|\delta(C)|\leq t\leq 2s$ and the recursion depth is at most $r$, this recursion is called at most $s^{O(r)}$ times.

\subsubsection{Dynamic Programming}\label{sec:app_dp}
Let $U$ be a finite set, $G=(V,E)$ be a directed acyclic graph, and $q \in V^U$ be a labeling of vertices such that no path in $G$ passes through multiple vertices with the same label.
For a path $P$ in $G$, let $q(P) \subseteq U$ be the set of labels of the vertices in $P$.
Let $\mathcal{D} \subseteq 2^U$ be the family of subsets $D \subseteq U$ such that there exists a longest path $P$ in $G$ with $q(P) = D$.
This formulation captures the solution domains of several typical dynamic programming problems.
In particular, the solution domain of the interval scheduling problem, for which Hanaka et al.~\cite{hanaka2021finding} provided an FPT algorithm for the max-min diversification problem parameterized by $k+\ell $, is represented as
\begin{itemize}
    \item $U = V$ is the set of intervals,
    \item $E$ is the set of pairs of intervals $(u,v)$ such that the right end of $u$ is to the left of the left end of $v$, and
    \item $q_v = v$ for all $v \in V$.
\end{itemize}
The solution domain of the longest common subsequence problem, for which Shida et al.~\cite{shida2024finding} provided an FPT algorithm for the max-min diversification problem parameterized by $k+\ell $ for the case the number $m$ of input strings $\{S_1, \dots, S_m\}$ is an absolute constant, is represented as
\begin{itemize}
    \item $U = \{1, \dots, \ell\} \times \Sigma$, where $\Sigma$ is the set of letters appearing in the input,
    \item $V = \{p := (p_1, \dots, p_m) \colon s_p := S_{1, p_1} = \dots = S_{m, p_m}\}$,
    \item $E$ is the set of pairs $(p, p')\in V\times V$ such that $p_i < p'_i$ for all $i \in \{1, \dots, m\}$, and
    \item $q_p = (\mathrm{len}_p, s_p)$ for all $p \in V$, where $\mathrm{len}_p$ represents the maximum number of vertices of a path that ends at $p$.
\end{itemize}
Let $k \in \mathbb{Z}_{\geq 1}$ and $d \in \mathbb{Z}_{\geq 0}$.
In this section, we design the $(-1,1)$-optimization oracle on $\mathcal{D}$ and the exact extension oracle on $\mathcal{D}$.
Particularly, by constructing these oracles, we obtain an FPT algorithm parameterized by $k+d$ for the max-min diversification problem on $\mathcal{D}$, which is a stronger result in terms of parameterization compared to the result in~\cite{hanaka2021finding} for the interval scheduling problem and the result in~\cite{shida2024finding} for the longest common subsequence problem.

We provide a $(-1,1)$-optimization oracle on $\mathcal{D}$ that runs in polynomial time. Let $w \in \{-1,1\}^U$.
This can be achieved by dynamic programming where for each $v \in V$, $\OPT[v]$ denotes the maximum number of vertices in a path ending at $v$, and $\DP[v]$ denotes the maximum total weight of the labels on a path ending at $v$ with $\OPT[v]$ vertices.

We provide an exact extension oracle on $\mathcal{D}$ that runs in polynomial time. Let $r \in \mathbb{Z}_{\geq 0}$, $X, Y \subseteq U$, and $C \subseteq U$.
This can be achieved by dynamic programming where for each $v \in V$, $p \in \{0, \dots, |X|\}$, $q \in \{0, \dots, r\}$, $\EX[v][p][q]$ denotes a path ending at $v$ with $\OPT[v]$ vertices that does not include any vertex with its label in $Y$, includes $p$ vertices with labels in $X$, and includes $q$ vertices with labels not belonging to $C$, if such a path exists; otherwise, it is $\bot$.

\bibliographystyle{plainurl}
\bibliography{bib}
\end{document}